%% file: rpaths-main.tex
\documentclass{article}
\usepackage{fullpage}
\usepackage{amsmath,amsthm}
\usepackage{url,afterpage}

\usepackage[utf8]{inputenc}
\usepackage{float}

\usepackage{tikz}
\usetikzlibrary{arrows}
\usepackage{tikz-network}
\usepackage{array, hyperref}
\usepackage{enumerate}
\usepackage{enumitem}

\usepackage[ruled]{algorithm}
\usepackage[commentColor=black]{algpseudocodex}

\newcommand{\mComment}[1]{\LComment{Comment: #1}}

\newtheorem{theorem}{Theorem}

\newtheorem{lemma}[theorem]{Lemma}
\newtheorem*{lemma*}{Lemma}

\newtheorem*{theorem*}{Theorem}
\newtheorem{definition}{Definition}

\theoremstyle{remark}
\newtheorem{note}{Note}

\title{Near Optimal Bounds for Replacement Paths and Related Problems in the CONGEST Model\footnote{Authors' affiliation: The University of Texas at Austin, USA; email:  {\tt vigneshm@cs.utexas.edu, vlr@cs.utexas.edu}. This work was supported in part by NSF grant CCF-2008241.}
\author {Vignesh Manoharan  and Vijaya Ramachandran}
}
\date{}

\begin{document}

\maketitle

\begin{abstract}
   We present several results in the CONGEST model on round complexity for  Replacement Paths (RPaths), Minimum Weight Cycle (MWC), and All Nodes Shortest Cycles (ANSC). We study these fundamental problems in both directed and undirected graphs, both weighted and unweighted. Many of our results are optimal to within a polylog factor: For an $n$-node graph $G$ we establish near linear lower and upper bounds for computing RPaths if $G$ is directed and weighted, and for computing MWC and ANSC if $G$ is weighted, directed or undirected; near $\sqrt{n}$ lower and upper bounds for undirected weighted RPaths; and $\Theta(D)$ bound for undirected unweighted RPaths. We also present lower and upper bounds for approximation versions of these problems, notably a $(2-(1/g))$-approximation algorithm for undirected unweighted MWC that runs in $\tilde{O}(\sqrt{n}+D)$ rounds, improving on the previous best bound of $\tilde{O}(\sqrt{ng}+D)$ rounds, where $g$ is the MWC length.
   We present a $(1+\epsilon)$-approximation algorithm for directed weighted RPaths, which beats the linear lower bound for exact RPaths.
\end{abstract}

\input{section-intro.tex}

\input{section-rp.tex}

\input{section-mwc.tex}

\input{section-path-constr.tex}

\input{conclusion.tex}

\bibliographystyle{plainurl}
\bibliography{main}

\end{document}

%% file: section-intro.tex
\section{Introduction}

Consider a communication network $G=(V,E)$, with two special nodes $s$ and $t$ in $V$, and with communication transmission from $s$ to $t$ along a shortest path $P_{st}$. 
In the distributed setting, it can be important to maintain efficient communication from $s$ to $t$ in the event that a link (i.e., edge) on this path $P_{st}$ fails. This is the Replacement Paths (RPaths) problem, where for each edge $e$ on $P_{st}$, we need to find a shortest path from $s$ to $t$ that avoids $e$. The RPaths problem has been extensively studied in the sequential setting~\cite{lawler1972procedure,katoh1982efficient,roditty2012replacement}. The closely related problem of finding a second simple shortest path (2-SiSP), i.e.,  a shortest simple path from $s$ to $t$ avoids at least one edge on the original shortest path $P_{st}$, has also been studied in the sequential setting~\cite{yen1971finding}. Another fundamental graph problem related to shortest paths is the Minimum Weight Cycle (MWC) problem, where we need to compute a shortest simple cycle in a given graph. The All Nodes Shortest Cycle (ANSC) problem, where we want to compute a shortest cycle through each node in a given graph, is also relevant to the distributed setting. 
Surprisingly, there are virtually no results known in the distributed CONGEST model for these problems. 
In this paper we address this lacuna by obtaining
strong round complexity bounds for these problems in the CONGEST model, which are optimal or near-optimal in many cases.

Let $|V|=n$, $|E|=m$.
Directed weighted RPaths, 2-SiSP, MWC, ANSC are important problems in sequential fine-grained complexity, being part of the $n^3$ time complexity class~\cite{williams2010subcubic} which contains many graph problems: weighted APSP, Negative Triangle Detection, Radius, Eccentricities and Betweenness Centrality~\cite{abboud2015subcubic}. Directed weighted MWC is the starting point of hardness in the fine-grained $mn$ complexity class~\cite{agarwal2018finegrained}, which contains directed weighted RPaths, 2-SiSP, ANSC. The $mn$ class is of more relevance in the CONGEST model where $O(\log n)$ bits of communication occur per round per edge in the graph, leading to $\tilde{O}(m)$\footnote{ We use the notation $\tilde{O}, \tilde{\Omega}, \tilde{\Theta}$ to hide poly-logarithmic factors.} communication per round, and the goal is to minimize the number of rounds in the computation.

To the best of our knowledge, the only previous work in the CONGEST model with a result for RPaths is a randomized algorithm for single-source fault-tolerant paths in undirected unweighted graphs~\cite{ghaffari2016fault}. We show optimal deterministic bounds for RPaths in this setting, and we obtain nontrivial, and in many cases near-optimal, results for RPaths in directed graphs and weighted graphs. 

Weighted directed RPaths and 2-SiSP are in the $mn$ and $n^3$ sequential fine-grained complexity classes as discussed above. But these two problems are unique among the problems in these sequential classes in that they become simpler for undirected graphs~\cite{katoh1982efficient} and for unweighted directed graphs~\cite{roditty2012replacement}. The results we present in this paper for RPaths and 2-SiSP for the CONGEST model show a similar trend, with the added contribution of an unconditional near linear lower bound for the weighted directed case along with sublinear round algorithms for unweighted and undirected graphs (as long as the network diameter and length of $P_{st}$ are sublinear). Thus we show a proven separation in complexity for RPaths and 2-SiSP between weighted directed graphs and either unweighted directed graphs or undirected graphs (even if weighted) in the CONGEST model.

Although RPaths and 2-SiSP have not been extensively studied in the CONGEST model, other related and more recently defined problems have been studied: CONGEST algorithms for fault-tolerant distance preservers~\cite{bodwin2023restorable,ghaffari2016fault}, which are sparse subgraphs on the network in which replacement path distances are exactly preserved, and fault-tolerant spanners~\cite{parter2022vertex,dinitz2020spanner}, which are sparse subgraphs in which replacement path distances are approximations of the distances in the original network have been reported (see Section~\ref{sec:prior} for details). The techniques in these results do not readily give efficient algorithms to explicitly compute replacement path weights or to construct a replacement path when a failed edge is known; further, these results mainly deal with undirected unweighted graphs. 

We also present results for computing MWC and ANSC in the CONGEST model. We present near-linear upper and lower bounds for the exact computation of MWC in directed graphs and weighted graphs.
The only variant of MWC previously studied in the CONGEST was in undirected unweighted graphs, where MWC is known as girth. The current best upper and lower bounds for exact computation in the CONGEST model are $O(n)$~\cite{holzer2012apsp} and $\tilde{\Omega}(\sqrt{n})$~\cite{frischknecht2012} respectively. For 2-approximation the previous best upper bound was $\tilde{O}(\sqrt{ng}+D)$~\cite{peleg2013girth} ($g$ is the weight of MWC), which we improve to $\tilde{O}(\sqrt{n}+D)$. We also present an algorithm for $(2+\epsilon)$-approximation of MWC in undirected graphs, and we substantially improve on this results in a recent paper~\cite{mwcarxiv} which studies the MWC problem in the CONGEST model in more detail. Additionally in that paper, we present lower bounds for arbitrary constant approximation and new sublinear algorithms for directed MWC. 

An extended abstract of our results on RPaths and 2-SiSP presented here appears in~\cite{rp2024}. Some of the results on MWC are included in~\cite{mwc2024}, which also has other results.

\subsection{Preliminaries}
\label{sec:prelim}
\noindent
\textbf{The CONGEST Model.}\label{congest-model}
In the CONGEST model~\cite{peleg2000distributed}, a communication network is represented by a graph $G=(V,E)$ where nodes model processors and edges model bounded-bandwidth communication links between processors. Each node has a unique identifier in  $\{0, 1, \dots n-1\}$ where $n = |V|$, and each node only knows the identifiers of itself and its neighbors in the network. Each node has infinite computational power. The nodes perform computation in synchronous rounds, where each node can send a message of up to $\Theta(\log n)$ bits to each neighbor and can receive the messages sent to it by its neighbors. The complexity of an algorithm is measured by the number of rounds until the algorithm terminates. 

We consider both weighted and unweighted graphs $G$ in this paper, where in the weighted case each edge has an integer weight which is known to the vertices incident to the edge. The graph $G$ can be directed or undirected. Following the convention for CONGEST algorithms~\cite{chechiksssp,nanongkai2014approx,ghaffari2015reach,agarwal2020deterministic,ancona2020}, the communication links are always bi-directional and unweighted.

\noindent
\textbf{Notation.}
Let $G=(V,E)$ be a directed or undirected graph with $|V|=n$ and $|E|=m$. Let edge $(u,v)$ have non-negative integer weight $w(u,v)$ according to a weight assignment function $w: E \rightarrow \{0,1, \dots W\}$, where $W = poly(n)$. Let $\delta_{st}$ denote the weight of a shortest path $P_{st}$ from $s$ to $t$ and $h_{st}$ denote the number of edges (hop distance) on this shortest path. The undirected diameter $D$ is the maximum shortest path distance between any two vertices in the underlying undirected unweighted graph of $G$. 

\begin{note}
\label{note:sssp}
    We use $SSSP$ and $APSP$ to denote the round complexity in the CONGEST model for weighted single source shortest paths {\rm (SSSP)} and weighted all pairs shortest paths {\rm (APSP)} respectively. The current best algorithm for weighted {\rm APSP} runs in $\tilde{O}(n)$ rounds, randomized~\cite{bernsteinapsp}. For weighted {\rm SSSP}, recent results~\cite{rozhon22sssp,cao2023sssp} provide an $\tilde{O}(\sqrt{n} + n^{2/5+o(1)}D^{2/5}+D)$ round randomized algorithm. The current best lower bounds are $\Omega\left(\frac{\sqrt{n}}{\log n} + D\right)$ for weighted {\rm SSSP}~\cite{peleg2000mst} and $\Omega\left(\frac{n}{\log n}\right)$ for (weighted and unweighted) {\rm APSP}~\cite{nanongkai2014approx}.
\end{note}

We now define the problems we consider in this paper.
\begin{definition}\label{def:rpaths}
    \textbf{Replacement Paths (RPaths)} : Given an $n$-node graph $G$, two vertices $s,t$ and a shortest path $P_{st}$ from $s$ to $t$, for each edge $e \in P_{st}$, compute the weight $d(s,t,e)$ of a shortest simple path $P_e$ from $s$ to $t$ that does not contain $e$.

    \textbf{Second Simple Shortest Path (2-SiSP)} : Given an $n$-node graph $G$, two vertices $s,t$ and a shortest path $P_{st}$ from $s$ to $t$, compute the weight $d_2(s,t)$ of a shortest simple path $P_2$ from $s$ to $t$ that differs from $P_{st}$.

    \textbf{Minimum Weight Cycle (MWC)} : Given an $n$-node graph $G$, compute the weight of a minimum weight simple cycle in $G$. For an unweighted graph, this is called the {\it girth} and is the length of a minimum hop cycle but we will continue to call it MWC.

    \textbf{All Nodes Shortest Cycle (ANSC)} : Given an $n$-node graph $G$, for each vertex $v$, compute the weight of a minimum weight simple cycle in $G$ that passes through $v$.
\end{definition}

Our lower bounds apply even when only one node in the graph is required to know the weights to be computed. In our RPaths or 2-SiSP algorithms, all vertices can learn the distances $d(s,t,e)$ or $d_2(s,t)$ using a simple broadcast in $O(h_{st}+D)$ rounds, which is within the round complexity bounds. In our MWC and ANSC algorithms, broadcasting takes $O(D)$ and $O(n+D)$ rounds respectively which is dominated by the rounds to compute weights.

In our RPaths and 2-SiSP results, we assume that the shortest path $P_{st}$ between the vertices $s,t$ is part of the input and that each vertex in the network knows the identities of $s$ and $t$, and the identities of vertices on $P_{st}$. The round bounds of our algorithms are unchanged if we are required to compute $P_{st}$ using known CONGEST algorithms for SSSP and broadcast the necessary information in $O(h_{st}+D)$ rounds.
Also, for computing 2-SiSP weight, note that once we have the weights of the $h_{st}$ replacement paths for $P_{st}$ we can compute the minimum among them in additional $O(D)$ rounds with a convergecast. Similarly, MWC can be computed in $O(D)$ rounds once ANSC values have been computed with a convergecast.

\subsection{Our Results}
\begin{table}
    \caption{{\it Exact Weight/Size results.} $SSSP$ and $APSP$ refer to the round complexity of weighted SSSP and APSP (See Note~\ref{note:sssp}). $^\dagger$Denotes deterministic results, all other results are randomized}
    \label{tab:results}
    \begin{minipage}{\columnwidth}
        \begin{center}
        \begin{tabular}{|c| c | c |}\hline
            \textbf{Problem} & \textbf{Lower Bound} & \textbf{Upper Bound} \\ \hline \hline
            & \multicolumn{2}{c|}{ \textbf{Directed Weighted Graphs}} \\ \hline
            \textit{RPaths} & $\Omega(APSP) = \Omega\left(\frac{n}{\log n}\right)$ [Thm \ref{thm:dirrp}.A] & $O(APSP) = \tilde{O}(n)$ [Thm \ref{thm:dirrp}.B]  \\ \hline
            \textit{MWC, ANSC} & $\Omega\left(\frac{n}{\log n}\right)$  [Thm \ref{thm:dirrp}.A] & $\tilde{O}(n)$ \cite{bernsteinapsp} \\ \hline
            & \multicolumn{2}{c|}{\textbf{Directed Unweighted Graphs}} \\ \hline
            \textit{RPaths} &  $\Omega\left(\frac{\sqrt{n}}{\log n} + D\right)$ [Thm \ref{thm:dirunwrp}.A] & $\tilde{O}(\min(n^{2/3} + \sqrt{nh_{st}} + D, SSSP \cdot h_{st}))$ [Thm \ref{thm:dirunwrp}.B] \\ \hline
            \textit{MWC, ANSC} &  $\Omega\left(\frac{n}{\log n}\right)$ [Thm \ref{thm:dirrp}.A] & $O(n)^\dagger$ \cite{holzer2012apsp} \\ \hline
            & \multicolumn{2}{c|}{ \textbf{Undirected Weighted Graphs}} \\ \hline
            \textit{RPaths} & $\Omega(SSSP) = \Omega\left(\frac{\sqrt{n}}{\log n} + D\right)$ [Thm \ref{thm:undirrp}.A] & $O(SSSP + h_{st}) = O(\sqrt{n}D^{1/4} + h_{st} + D)$ [Thm \ref{thm:undirrp}.B] \\ \hline
            \textit{MWC, ANSC} &  $\Omega\left(\frac{n}{\log n}\right)$ [Thm \ref{thm:undirmwc}.A] & $\tilde{O}(n)$  [Thm \ref{thm:undirmwc}.B] \\ \hline
            & \multicolumn{2}{c|}{\textbf{Undirected Unweighted Graphs}} \\ \hline
            \textit{RPaths} & $\Omega\left(D\right)$ [Thm \ref{thm:dirunwrp}.A] & $O(D)^\dagger$ [Thm \ref{thm:dirunwrp}.B] \\ \hline
            \textit{MWC, ANSC} &  $\Omega\left(\frac{\sqrt{n}}{\log n}\right)$ (adapted from~\cite{korhonen2017,drucker2014}) & $O(n)^\dagger$  \cite{holzer2012apsp} \\ \hline
        \end{tabular}
    \end{center}
    \end{minipage}
\end{table}

The upper and lower bounds proved in this paper are tabulated in Table~\ref{tab:results}. We use the notation $\tilde{O}, \tilde{\Omega}, \tilde{\Theta}$ to hide poly-logarithmic factors. We also introduce the notation $\tilde{o}$, where $\tilde{o}(f(n))$ denotes the class of functions that are asymptotically smaller than $f(n)$ by at least a polynomial factor\footnote{Our results apply if we alternatively define $\tilde{o}(f(n))$ as the class of functions asymptotically smaller than $f(n)$ by a factor of $\omega(\log^k n)$ for any constant $k>1$.}. 

We also consider approximation versions of the problems defined above, these results are tabulated in Table~\ref{tab:approxresults}. The $\alpha$-approximation (for some $\alpha > 1$) for the RPaths problem is to compute a replacement path for each edge $e$ whose weight is within an $\alpha$ factor of the shortest replacement path for $e$, and the other approximation problems are defined similarly. We also use the notation $(1+\epsilon)$-approximation where $\alpha = (1+\epsilon)$ for an arbitrarily small constant $\epsilon > 0$. 

Our algorithms compute the weight/size of the replacement path or minimum weight cycle, and we also show how they can be modified to construct the actual replacement paths or cycle. The construction can be done using routing tables at each vertex. In the case of undirected graphs, we present an additional on-the-fly method where the replacement path or cycle is quickly constructed when the failing edge is discovered without storing the entire routing table.
This is discussed in detail in Section~\ref{sec:recon}. We also consider the problem of \textit{$q$-Cycle Detection}: given an unweighted graph $G$, some vertex must report a cycle if $G$ contains a simple cycle of length $q$. If $G$ has no simple cycle of length $q$, no vertex should report a cycle. We now present our results.

\begin{table}
    \caption{{\it Approximate Weight/Size results.}
    The results hold for an arbitrarily large constant $\alpha > 1$ and arbitrarily small constant $\epsilon > 0$. All results are randomized.}
    \label{tab:approxresults}
    \begin{minipage}{\columnwidth}
        \begin{center}
        \begin{tabular}{|c| c | >{\centering\arraybackslash}p{8cm} |}\hline
            \textbf{Problem} & \textbf{Lower Bound} & \textbf{Upper Bound} \\ \hline \hline
            & \multicolumn{2}{c|}{ \textbf{Directed Weighted Graphs}} \\ \hline
            \textit{RPaths} &  1, $\Omega\left(\frac{n}{\log n}\right)$ & $(1+\epsilon)$,  $\tilde{O}(n^{2/3} + \sqrt{nh_{st}} + D)$ [Thm \ref{thm:dirrp}.C]  \\ \hline
            \textit{MWC, ANSC} & $(2-\epsilon), \Omega\left(\frac{n}{\log n}\right)$  & 1, $\tilde{O}(n)$ \\ \hline
            & \multicolumn{2}{c|}{\textbf{Directed Unweighted Graphs}} \\ \hline
            \textit{RPaths} &  $\alpha, \Omega\left(\frac{\sqrt{n}}{\log n} + D\right)$ [Thm \ref{thm:dirrp}.A] & 1, $\tilde{O}(n^{2/3} + \sqrt{nh_{st}} + D)$  \\ \hline
            \textit{MWC, ANSC} &  $(2-\epsilon), \Omega\left(\frac{n}{\log n}\right)$ [Thm \ref{thm:dirrp}.A] & 1, ${O}(n)$  \\ \hline
            & \multicolumn{2}{c|}{ \textbf{Undirected Weighted Graphs}} \\ \hline
            \textit{RPaths} & $\alpha, \Omega\left(\frac{\sqrt{n}}{\log n} + D\right)$ [Thm \ref{thm:dirunwrp}.A] & 1, $O(SSSP + h_{st})$ \\ \hline
            \textit{MWC, ANSC} &  $(2-\epsilon), \Omega\left(\frac{n}{\log n}\right)$ [Thm \ref{thm:undirmwc}.A] & $(2+\epsilon)$, $\tilde{O}\big(\min(n^{3/4}D^{1/4} + n^{1/4}D, $ $n^{3/4} + n^{0.65}D^{2/5} + n^{1/4}D, n)\big)$ [Thm \ref{thm:undirmwc}.D] \\ \hline
            & \multicolumn{2}{c|}{\textbf{Undirected Unweighted Graphs}} \\ \hline
            \textit{RPaths} & 1, $\Omega\left(D\right)$ & 1, $O(D)$ \\ \hline
            \textit{MWC} &  $(2-\epsilon), \Omega\left(\frac{\sqrt{n}}{\log n}\right)$ \cite{frischknecht2012} & $(2-(1/g))$,$\tilde{O}(\sqrt{n} + D)$ [Thm \ref{thm:undirmwc}.C] \\ \hline
        \end{tabular}
    \end{center}
    \end{minipage}
\end{table}

\vspace{0.5em}
\noindent
\textbf{Directed Weighted Graphs.}
For an $n$-node directed weighted graph, we present an RPaths CONGEST algorithm that runs in near-linear $\tilde{O}(n)$ rounds (Section~\ref{sec:dirrpub}). The classic sequential $\tilde{O}(mn)$-time algorithm for 2-SiSP and RPaths~\cite{yen1971finding} performs a sequence of $h_{st}$ SSSP computations, and a near-linear bound is not achievable on the CONGEST model through implementing this algorithm. Instead, we formulate RPaths as an APSP computation (on an alternate graph) that can be efficiently computed within the APSP bound $\tilde{O}(n)$ in the CONGEST model. We show that our algorithm is nearly optimal by presenting an $\tilde{\Omega}(n)$ lower bound (even when the undirected diameter $D$ is a constant) for both RPaths and 2-SiSP through a reduction from set disjointness (Section~\ref{sec:dirrplb}). 

Our lower bound proof for RPaths is much more involved than the $\tilde{\Omega}(n)$ APSP lower bound in~\cite{nanongkai2014approx}, and we also 
show that RPaths differs from APSP in efficient approximability: the $\tilde{\Omega}(n)$ APSP lower bound in~\cite{nanongkai2014approx} applies to $\alpha$-approximation for any constant $\alpha > 1$,  but for weighted directed RPaths we give in Section~\ref{sec:approxdirrpub} an asymptotically improved algorithm for $(1 + \epsilon)$-approximation (for any constant $\epsilon > 0$) that runs in sublinear rounds ($\tilde{O}(n^{1-c})$ for constant $c > 0$) whenever both $h_{st}$ and $D$ are sublinear.

\begin{theorem}
    \label{thm:dirrp}
    Given a directed weighted graph $G$ on $n$ vertices with undirected diameter $D$ and a shortest path $P_{st}$ of hop length $h_{st}$,
    \begin{enumerate}[label=\Alph*.,ref=\Alph*]
        \item \label{thm:dirrp:lb} Any randomized algorithm that computes RPaths or 2-SiSP for $P_{st}$ requires $\Omega\left(\frac{n}{\log n}\right)$ rounds, even if $D$ is constant.
        \item \label{thm:dirrp:ub}  RPaths and 2-SiSP for $P_{st}$ can be computed in $O(APSP)$ rounds, and hence by a randomized algorithm in $\tilde{O}(n)$ rounds.
        \item \label{thm:dirrp:approxub} There is a randomized algorithm that computes a $(1+ \epsilon)$-approximation of RPaths and 2-SiSP in $\tilde{O}(\sqrt{nh_{st}} + D + \min(n^{2/3}, h_{st}^{2/5}n^{2/5+o(1)}D^{2/5}))$ rounds, for any constant $\epsilon >0$.
        \item \label{thm:dirrp:approxlb}  Computing an $\alpha$-approximation of RPaths or 2-SiSP for any constant $\alpha >1$ requires $\Omega\left(\frac{\sqrt{n}}{\log n} + D\right)$ rounds.
    \end{enumerate}
\end{theorem}

In directed graphs (weighted or unweighted), we can compute exact values of all nodes shortest cycle (ANSC) in the CONGEST model by first computing APSP~\cite{holzer2012apsp,bernsteinapsp} and then computing at each vertex $v \in V$ the minimum among cycles formed by concatenating a $v$-$u$ shortest path and a single edge $(u,v)$ (for each in-neighbor $u$ of $v$). Minimum weight cycle (MWC) can be computed after computing ANSC in $O(D)$ additional rounds by a convergecast operation computing the minimum among cycles at each vertex. Since we can compute APSP in $\tilde{O}(n)$ rounds in directed weighted graphs~\cite{bernsteinapsp} and $O(n)$ time in directed unweighted graphs~\cite{holzer2012apsp}, we obtain CONGEST algorithms for exact ANSC and MWC with these same round complexities (see Section~\ref{sec:mwcub}). Matching these upper bounds, we present near linear lower bounds for MWC and ANSC in directed graphs (Section~\ref{sec:dirmwclb}). We prove these lower bounds using reductions from set disjointness, and the lower bounds apply even to $(2-\epsilon)$-approximation.
In a recent paper~\cite{mwcarxiv}, we present $\tilde{\Omega}(\sqrt{n})$ lower bounds for arbitrarily large constant approximation of directed MWC, and sublinear $\tilde{O}(n^{4/5}+D)$ algorithms for $(2+\epsilon)$-approximation of MWC (2-approximation in unweighted graphs).

\begin{theorem}
    \label{thm:dirmwc}
    Given a directed weighted~(or unweighted) graph $G$ on $n$ vertices with undirected diameter $D$, any algorithm that computes MWC or ANSC requires $\Omega\left(\frac{n}{\log n}\right)$ rounds, even if $D$ is constant. This lower bound also applies to any $(2-\epsilon)$-approximation algorithm for MWC. 
\end{theorem}

\vspace{0.5em}
\noindent
\textbf{Directed Unweighted Graphs.}
In the case of directed unweighted graphs, the near linear lower bound for the weighted case no longer applies. We give an algorithm based on sampling and computing detours that takes $\tilde{O}(n^{2/3} + \sqrt{nh_{st}} + D)$ rounds. This gives us an algorithm that runs in sublinear rounds whenever both $h_{st}$ and $D$ are sublinear. We also have a simple algorithm taking $O(h_{st} \cdot SSSP)$ rounds that is more efficient when $h_{st}$ is small (Section~\ref{sec:dirunwrpub}). We show a lower bound of ${\Omega}(SSSP) = \tilde{\Omega}(\sqrt{n}+D)$ for computing RPaths and 2-SiSP on directed unweighted graphs (Section~\ref{sec:dirunwrplb}), and our algorithm matches this $O(SSSP)$ bound when $h_{st}$ is $O(1)$. This $\tilde{\Omega}(\sqrt{n}+D)$  lower bound shows that computing RPaths is harder in directed unweighted graphs than in undirected unweighted graphs, where we have an $O(D)$ round algorithm (see below).

\begin{theorem}
    \label{thm:dirunwrp}
    Given a directed unweighted graph $G$ on $n$ vertices with undirected diameter $D$ and a shortest path $P_{st}$ of hop length $h_{st}$,
    \begin{enumerate}[label=\Alph*.,ref=\Alph*]
        \item \label{thm:dirunwrp:lb}Any randomized algorithm that computes RPaths or 2-SiSP requires $\Omega\left(\frac{\sqrt{n}}{\log n} + D\right)$ rounds, even if $h_{st}$ and $D$ are as small as $\Theta(\log n)$. These lower bounds also apply to an $\alpha$-approximation, for any constant $\alpha >1$.
        \item \label{thm:dirunwrp:ub}There is a randomized algorithm that computes RPaths and 2-SiSP for $P_{st}$ in $\tilde{O}(\min(n^{2/3} + \sqrt{nh_{st}} + D, SSSP \cdot h_{st}))$ rounds.
    \end{enumerate}
\end{theorem}

We adapt our lower bound proof for directed unweighted RPaths to other basic directed graph problems such as $s$-$t$ reachability and $s$-$t$ shortest path in directed unweighted graphs. A folklore lower bound of $\tilde{\Omega}(\sqrt{n} + D)$ for these problems was attributed by~\cite{ghaffari2015reach} to {\it undirected} lower bound results in~\cite{sarma2012distributed}, and we give explicit proofs here. A lower bound of $\tilde{\Omega}(\sqrt{n} + D)$ for undirected weighted SSSP was known~\cite{elkin2006mst,sarma2012distributed} (which also applies to directed weighted SSSP) but this does not apply to unweighted directed graphs. These problems are easier in undirected unweighted graphs since undirected BFS can be performed in $O(D)$ rounds.
Thus our results indicate that basic problems in directed graphs are asymptotically harder than their undirected unweighted counterparts.
We note that this difference is not very surprising since the underlying communication network is the undirected version of the graph regardless of whether the graph is directed or undirected.

We also prove a $\Omega(n/\log n)$ lower bound for directed fixed-length cycle detection for length $q\geq 4$. This significantly differs from undirected fixed-length cycle detection, where we only have a lower bound of $\Omega(\sqrt{n}/\log n)$ and we have a $\tilde{O} (\sqrt n)$ upper bound for detecting an undirected 4-cycle. These connections are discussed further in Section~\ref{sec:cycledet}.

\begin{theorem}
    \label{thm:dirunwother}
    Given a directed unweighted graph $G$ on $n$ vertices,
    \begin{enumerate}[label=\Alph*.,ref=\Alph*]
        \item \label{thm:dirunwother:st}Any CONGEST algorithm that checks $s$-$t$ reachability, or computes $s$-$t$ shortest path in a directed unweighted graph requires $\Omega\left(\frac{\sqrt{n}}{\log n} + D\right)$ rounds.
        \item \label{thm:dirunwother:cycledet}Any CONGEST algorithm that detects a directed cycle of length $q$ (for any constant $q \ge 4$) in a directed unweighted graph on $n$ vertices requires $\Omega(\frac{n}{\log n})$ rounds.
    \end{enumerate}
\end{theorem}

\vspace{0.5em}
\noindent
\textbf{Undirected Graphs}.
For undirected graphs, our upper and lower bounds for RPaths match the round complexity of SSSP (BFS for unweighted) in the CONGEST model, except for weighted RPaths which requires an additional $O(h_{st})$ rounds. 
The remaining gap between our upper and lower bounds is inherited from the gap between the current best bounds for SSSP. 

\begin{theorem}\label{thm:undirrp}
    Given an undirected weighted~(or unweighted) graph $G$ on $n$ vertices with diameter $D$ and a shortest path $P_{st}$ of hop length $h_{st}$,
    \begin{enumerate}[label=\Alph*.,ref=\Alph*]
        \item\label{thm:undirrp:lb} Any algorithm that computes RPaths or 2-SiSP requires:
            \begin{enumerate}[label=\roman*.]
                \item $\Omega(SSSP) = \Omega\left(\frac{\sqrt{n}}{\log n} + D\right)$ rounds if $G$ is weighted, even if $h_{st}$ is constant. This lower bound applies to $\alpha$-approximation, for any $\alpha >1$.
                \item $\Omega\left(D\right)$ rounds if $G$ is unweighted, which is a tight bound.
            \end{enumerate}
        \item\label{thm:undirrp:ub} We can compute RPaths for $P_{st}$ in $O(SSSP + h_{st}) =  \tilde{O}(\sqrt{n} + n^{2/5+o(1)}D^{2/5}+D + h_{st})$ rounds. For 2-SiSP the bound is $O(SSSP)$. If $G$ is unweighted, the bound is $O(D)$ rounds.
    \end{enumerate}.
\end{theorem}

For undirected MWC, we prove a near linear lower bound for weighted graphs, giving us tight upper and lower bounds (Section~\ref{sec:mwclb}). For undirected unweighted graphs, we adapt the lower bound in~\cite{korhonen2017,drucker2014} for detecting copies of 6-cycles in a graph to give a $\Omega(\sqrt{n}/\log n)$ lower bound for MWC --- this lower bound also applies to any $(2-\epsilon)$ approximation (matching the result in~\cite{frischknecht2012}). We also give a $(2- (1/g))$-approximation algorithm for  undirected unweighted MWC running in $\tilde{O}(\sqrt{n}+D)$ rounds (Section~\ref{sec:approxundirmwcub}), which significantly improves on the $(2- (1/g))$-approximation algorithm of~\cite{peleg2013girth} which has round complexity $\tilde{O}(\sqrt{ng}+D)$; here $g$ is the length of the MWC. We also give a $(2+\epsilon)$-approximation algorithm for undirected weighted MWC, which has sublinear round complexity when $D$ is $\tilde{o}(n^{3/4})$. Our recent paper~\cite{mwcarxiv} presents additional results for approximation, including a lower bound of $\tilde{\Omega}(\sqrt{n})$ for arbitrarily large constant approximation, and a sublinear $\tilde{O}(n^{2/3}+D)$ algorithm for $(2+\epsilon)$-approximate undirected weighted MWC that significantly improves on the result presented in this paper.

\begin{theorem}
    \label{thm:undirmwc}
    Given a directed unweighted graph $G$ on $n$ vertices,
    \begin{enumerate}[label=\Alph*.,ref=\Alph*]
        \item \label{thm:undirmwc:lb} Any algorithm that computes the MWC or ANSC requires $\Omega\left(\frac{n}{\log n}\right)$ rounds, even if $D$ is constant. This lower bound also applies to any $(2-\epsilon)$-approximation algorithm for MWC.
        \item \label{thm:undirmwc:exact} There is an algorithm that solves MWC and ANSC in $O(APSP + n)= \tilde{O}(n)$ rounds. If $G$ is unweighted, there is an algorithm that takes $O(n)$ rounds.
        \item\footnote[2]{Detailed pseudocode and analysis of this algorithm is presented in our recent paper~\cite{mwcarxiv}. } \label{thm:undirmwc:approxunw} There is an algorithm that computes a $(2- (1/g))$-approximation of the MWC in an undirected unweighted graph in $\tilde{O}(\sqrt{n}+D)$ rounds.
        \item\footnote[3]{A significantly more efficient algorithm that takes $\tilde{O}(n^{2/3}+D)$ rounds is presented in our recent paper~\cite{mwcarxiv}. } \label{thm:undirmwc:approx} There is an algorithm that computes a $(2+\epsilon)$-approximation of the MWC in an undirected weighted graph in $\tilde{O}\left(\min(n^{3/4}D^{1/4} + n^{1/4}D, n^{3/4} + n^{0.65}D^{2/5} + n^{1/4}D, n)\right)$ rounds.
    \end{enumerate}
\end{theorem}

Our results indicate that  replacement paths can be computed faster than minimum weight cycle in many settings, and
RPaths is never harder to compute than MWC in the CONGEST model: 
Both problems have the same almost tight linear round bound for directed weighted graphs, but RPaths can be computed in $\tilde{o}(n)$ rounds  (when $h_{st}$ and $D$ are $\tilde{o}(n)$) while MWC has a lower bound of $\tilde{\Omega}(n)$ in both directed unweighted graphs and undirected weighted graphs. In undirected unweighted graphs, we have a tight $O(D)$ upper bound for RPaths while MWC has a lower bound of $\tilde{\Omega}(\sqrt{n})$.

\subsection{Prior Work}
\label{sec:prior}
RPaths and the closely related 2-SiSP are well-studied problems in the sequential setting. For weighted directed graphs the classical algorithm of Yen~\cite{yen1971finding} runs in~$\tilde{O}(mn)$ time and has a matching fine-grained lower bound of $\tilde{\Omega}(mn)$ assuming a sequential hardness result for MWC~\cite{agarwal2018finegrained}. For unweighted directed graphs, a randomized $\tilde{O}(m\sqrt{n})$ algorithm is given in~\cite{roditty2012replacement}. For undirected graphs, a near-linear time algorithm is given in~\cite{katoh1982efficient}, matching the running time for sequential SSSP. Our bounds for RPaths for these different graph classes in the CONGEST model follow a similar pattern: close to APSP for weighted directed graphs, close to SSSP for undirected graphs, and intermediate bounds for directed unweighted graphs. The more general problem of single source replacement paths (SSRP) was studied in the sequential setting in~\cite{chechik2020ssrp,chechik2019ssrp}.

The problem of computing MWC has been extensively studied in the sequential setting. MWC can be computed by computing All Pairs Shortest Paths (APSP) in the given graph in $O(n^3)$ time and in $\tilde{O}(mn)$ time. The hardness of computing MWC in the sequential fine-grained setting was shown by~\cite{williams2010subcubic} for the $n^3$ class and MWC $mn$-hardness is the hypothesis used for establishing hardness for the $mn$ class~\cite{agarwal2018finegrained}. For additional prior work on sequential MWC, see~\cite{mwcarxiv}.

In the distributed setting, an $O(D\log n)$ algorithm for single source replacement paths in undirected unweighted graphs was given in~\cite{ghaffari2016fault}. We are not aware of any prior results in the CONGEST model for RPaths or 2-SiSP for directed graphs or for weighted graphs. Distributed constructions of fault-tolerant preservers, which construct a sparse subgraph that exactly preserves replacement path distances, have been studied in~\cite{bodwin2023restorable,parter2020fault} but their constructions do not give an efficient procedure to compute replacement path distances. Fault-tolerant spanners, which construct a sparse subgraph that approximates replacement path distances, have also been studied in CONGEST~\cite{dinitz2020spanner,parter2022vertex}. 

An $O(n)$ algorithm for computing girth was given in~\cite{holzer2012apsp}, and a $\tilde{\Omega}(\sqrt{n})$ lower bound for computing girth was given in \cite{frischknecht2012} which applies to any $(2-\epsilon)$-approximation algorithm. An $\tilde{O}(\sqrt{ng}+D)$-round algorithm was given in~\cite{peleg2013girth} to compute $(2-\frac{1}{g})$-approximation of girth(where $g$ is the girth).  Cycle detection has been studied with both upper and lower bounds for undirected graphs~\cite{fraigniaud2024evencycle,eden2021sublinear,censorhillel2020girth,drucker2014}. Tight bounds of $\tilde{\Theta}(n^{1/3})$ are known for triangle detection in undirected and directed graphs~\cite{chang2021near,Pettie2022,izumi2017triangle}. 

There are unconditional $\tilde{\Omega}(n)$ CONGEST lower bounds for several graph problems in the sequential $n^3$ and $mn$ fine-grained complexity class: APSP~\cite{nanongkai2014approx}, diameter, radius~\cite{abboud2016,ancona2020}. CONGEST algorithms for these problems have also been studied: APSP~\cite{bernsteinapsp,lenzen2019distributed,agarwal2020deterministic}, diameter and radius~\cite{abboud2016,ancona2020} and betweenness centrality~\cite{hoang2019round}.

The round complexity of both exact and approximate SSSP has been extensively researched~\cite{elkin2006mst,nanongkai2014approx,forster2018sssp,chechiksssp,cao2021approximatesssp,cao2023sssp}. For exact or $(1+\epsilon)$-approximate SSSP, the current best upper and lower bounds are $\tilde{O}(n^{2/5+o(1)}D^{2/5} + \sqrt{n} + D)$~\cite{cao2023sssp} and $\Omega(\sqrt{n}+D)$~\cite{elkin2006mst,sarma2012distributed} respectively. 
 
\subsection{Our Techniques}

\noindent
{\bf CONGEST Lower Bounds.}
Our lower bounds use reductions either from Set Disjointness or from other graph problems with known CONGEST lower bounds. Set Disjointness is a two party communication complexity problem, where two players Alice and Bob are each given a $k$-bit string $S_a$ and $S_b$ respectively. Alice and Bob need to communicate and decide if the sets represented by $S_a$ and $S_b$ are disjoint, i.e., whether there is no bit position $i$, $1 \le i \le k$ with $S_a[i] = 1$ and $S_b[i] = 1$. A classical result in communication complexity is that Alice and Bob must exchange $\Omega(k)$ bits even if they are allowed shared randomness~\cite{kushilevitzcomm,razborov1992,baryossef2002}. Lower bounds using such a reduction also hold against randomized algorithms.

Our reduction from Set Disjointness for the $\tilde{\Omega}(n)$ CONGEST lower bound for directed weighted RPaths is loosely inspired by a construction in a sparse sequential reduction from MWC to RPaths in~\cite{agarwal2018finegrained}. Our near linear lower bounds for directed MWC and undirected weighted MWC also use new constructions to obtain reductions from set disjointness. 
Some of our lower bounds use known unconditional CONGEST lower bounds for problems like $s$-$t$ Subgraph Connectivity and weighted $s$-$t$ Shortest Path~\cite{sarma2012distributed}.

\noindent
{\bf CONGEST Algorithms.}
Adapting the sequential algorithm for directed weighted replacement paths~\cite{yen1971finding} directly to the CONGEST model requires up to $n$ SSSP computations, which is inefficient.  Instead, our CONGEST algorithm builds on a sequential sparse reduction from RPaths to Eccentricities in~\cite{agarwal2018finegrained} which we tailor to work efficiently in the CONGEST model using weighted APSP~\cite{bernsteinapsp} as a subroutine. 
Our algorithm for directed unweighted RPaths is loosely based on the sequential algorithm in~\cite{roditty2012replacement}, but we make significant changes to obtain efficiency in the distributed setting. We use a variety of techniques such as sampling, computing shortest paths in skeleton graphs, and pipelined BFS.

We use sampling in conjunction with source detection in our  
$\tilde{O}(\sqrt{n}+D)$-round $(2- (1/g))$-approximation algorithm for undirected unweighted MWC.
The starting point of our algorithm is the randomized $(2- (1/g))$-approximation algorithm in~\cite{peleg2013girth} that runs in $\tilde{O}(\sqrt{ng}+D)$ rounds in a graph with MWC length $g$. We significantly improve this round complexity by removing the dependence on $g$, and our algorithm compares favorably to the $\tilde{\Omega}(\sqrt{n})$ lower bound for computing a $(2-\epsilon)$-approximation. Using this unweighted algorithm, along with a weight scaling technique and sampling, we present a $(2+\epsilon)$-approximation algorithm for undirected weighted MWC which has sublinear round complexity when $D$ is $\tilde{o}(n^{3/4})$. 

Many of our algorithms do not use any randomness apart from that used by the randomized algorithms for SSSP and APSP. If we use deterministic CONGEST algorithms for these problems, such as for unweighted APSP~\cite{lenzen2019distributed} and weighted APSP~\cite{agarwal2020deterministic}, our algorithms will be deterministic as well. The exceptions to this are our directed unweighted RPaths algorithm and approximate directed weighted RPaths, and our undirected MWC approximation algorithms, which inherently use random sampling. Our exact unweighted MWC and undirected unweighted RPaths algorithms are deterministic.

\subsubsection{Roadmap}
In Section~\ref{sec:prelim}, we define the CONGEST model and notation used in this paper. In Section~\ref{sec:rp}, we discuss lower and upper bounds for the RPaths and 2-SiSP problems. We give new lower bounds for computing RPaths in directed weighted~(Section~\ref{sec:dirrplb}), directed unweighted~(Section~\ref{sec:dirunwrplb}) and undirected~(Section~\ref{sec:undirrplb}) graphs. We also discuss lower bounds for other graph problems in directed unweighted graphs in Section~\ref{sec:dirunwother}. We then give algorithms for computing RPaths in directed weighted~(Section~\ref{sec:dirrpub}), directed unweighted~(Section~\ref{sec:dirunwrpub}) and undirected~(Section~\ref{sec:undirrpub}) graphs. We also discuss an approximation algorithm for directed weighted RPaths~(Section \ref{sec:approxdirrpub}). 

In Section~\ref{sec:mwc}, we discuss lower and upper bounds for the MWC and ANSC problems. We give new near-linear lower bounds for directed~(Section~\ref{sec:dirmwclb}) and undirected~(Section~\ref{sec:undirmwclb}) graphs, including approximation hardness results. We then discuss upper bounds for computing exact MWC, ANSC in Section~\ref{sec:mwcub}. We then present approximation algorithms for undirected unweighted~(Section~\ref{sec:approxundirmwcub}) and weighted~(Section~\ref{sec:approxundirwtmwcub}) MWC. We discuss cycle detection and present a directed cycle detection lower bound in Section~\ref{sec:cycledet}. In Section~\ref{sec:recon}, we discuss how to modify our algorithms to construct the replacement paths and minimum weight cycles. We conclude with some open problems in Section~\ref{conclusion}.

%% file: section-rp.tex
\section{Replacement Paths}
\label{sec:rp}

\subsection{Lower Bounds For Replacement Paths}

In Section~\ref{sec:dirrplb}, we prove near-linear unconditional CONGEST lower bounds for directed weighted RPaths and 2-SiSP. 
In Section~\ref{sec:dirunwrplb} we show a lower bound of $\tilde{\Omega}(\sqrt{n}+D)$  for computing RPaths and 2-SiSP in directed unweighted graphs, which also proves a folklore lower bound for directed single source reachability.

\subsubsection{Directed Weighted Replacement Paths}
\label{sec:dirrplb}

\begin{figure}[t]
    \centering
    \includegraphics[scale=0.65]{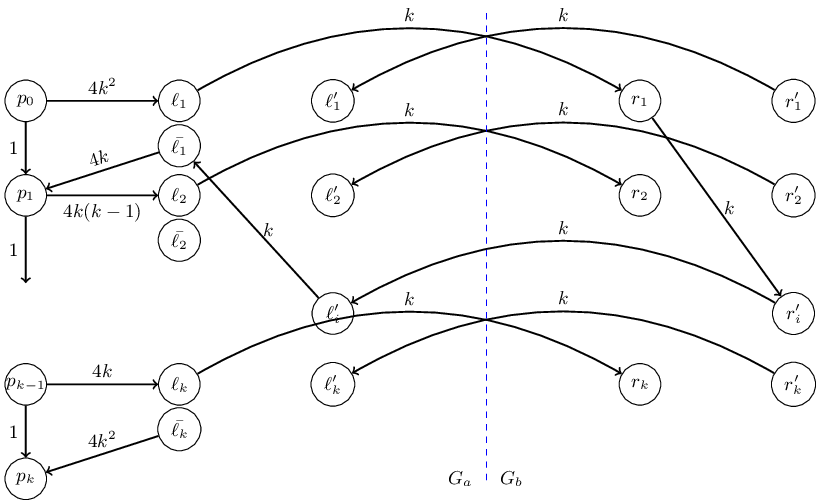}
    \caption{Directed weighted RPaths, 2-SiSP lower bound construction}
    \label{fig:dirrplb} 
\end{figure}

We prove Theorem~\ref{thm:dirrp}.\ref{thm:dirrp:lb} by showing an unconditional $\Omega\left(\frac{n}{\log n}\right)$ lower bound for computing 2-SiSP (which extends to RPaths) in directed weighted graphs. Our proof is based on a new reduction from set disjointness pictured in Figure~\ref{fig:dirrplb}, partly inspired by a sparse sequential reduction from MWC to RPaths in~\cite{agarwal2018finegrained}.

Consider an instance of the Set Disjointness problem where the players Alice and Bob are given $k^2$-bit strings $S_a$ and $S_b$ respectively representing sets of at most $k^2$ elements (the $i$'th bit being 1 indicates element $i$ belongs to the set). The problem is to determine whether $S_a \cap S_b = \phi$, i.e., that for all indices $1 \le i \le k^2$ either $S_a[i] = 0$ or $S_b[i]=0$. A classical result in communication complexity is that Alice and Bob must exchange $\Omega(k)$ bits even if they are allowed shared randomness~\cite{kushilevitzcomm,razborov1992,baryossef2002}. Our reduction constructs the graph $G=(V,E)$ described below and we show in Lemma~\ref{lem:dirrplb} that $G$ has a low-weight 2-SiSP if and only if the sets $S_a$ and $S_b$ are not disjoint.

We will construct $G=(V,E)$ with six sets of vertices (see Figure~\ref{fig:dirrplb}): $L = \{\ell_i \mid 1 \le i \le k\}, L' = \{\ell_i' \mid 1 \le i \le k\}, R = \{r_i \mid 1 \le i \le k\}, R' = \{r_i' \mid 1 \le i \le k\}$, $\overline{L} = \{\overline{\ell_i} \mid 1 \le i \le k\}$, $P = \{p_i \mid 0 \le i \le k\}$. Note that the number of vertices is $n = 6k+1$. We set $s=p_0$ and $t=p_k$ and for each $1 \le i \le k$, we add the edges $(p_{i-1}, p_{i})$ with weight 1; this is the input shortest path $\mathcal{P}=P_{st}$. We add directed edges $(\ell_i, r_i)$ and $(r_i', \ell_i')$ for each $1 \le i \le k$. Each of these edges has weight $k$. For each $1 \le i \le k$, we add the edges  $(p_{i-1},\ell_{i})$ with weight $4k(k-i+1)$ and $(\overline{\ell}_{i}, p_{i})$ with weight $4ki$. This is our base graph.

We now add edges to $G$ based on the set disjointness inputs $S_a$,$S_b$. We encode each integer $q$, $1 \leq q \leq k^2$, as an ordered pair $(i,j)$ such that $q= (i-1) \cdot k + j$. If 
$S_a[q]=1$,
we add the edge $(\ell_j', \overline{\ell_i})$ with weight $k$, if 
$S_b[q]=1$,
we add the edge $(r_i, r_j')$ with weight $k$. For the 2-SiSP problem, the desired output is $d_2(p_0, p_k)$, the weight of a second simple shortest path from $p_0$ to $p_k$.

\begin{lemma}
    If $S_a \cap S_b \ne \emptyset$, then $d_2(p_0, p_k) \le (4k^2 + 9k-1)$. Otherwise, if $S_a \cap S_b = \emptyset$, then $d_2(p_0, p_k) \ge (4k^2 + 12k)$.
    \label{lem:dirrplb}
\end{lemma}
\begin{proof}
    If the sets $S_a, S_b$ are not disjoint, then there exists $1\le i,j \le k$ such that $S_a[(i-1) \cdot k+j] = S_b[(i-1)\cdot k +j] = 1$. Then, the path $\langle p_{i-1}, \ell_i, r_i, r_j', \ell_i', \overline{\ell}_i, p_{i}\rangle$ provides a detour for the edge $(p_{i-1}, p_{i})$. This can be used along with the shortest paths from $p_0$ to $p_{i-1}$ and $p_{i}$ to $p_{k}$ to obtain a simple path of weight $4k(k+1)+4k+k-1$ that does not use edge $(p_{i-1}, p_i)$. So, the second simple shortest path from $p_0$ to $p_k$ has weight at most $4k^2+9k-1$.

    Assume the strings are disjoint. Let $\mathcal{P}_2$ be a second simple shortest path, and let $(p_{i-1},p_{i})$ be the first edge that is not in $\mathcal{P}_2$ but is in the $p_0$-$p_k$ shortest path $\mathcal{P}$. Since the only other outgoing edge from $p_{i-1}$ is $(p_{i-1}, \ell _i)$ (with weight $4k(k-i+1)$), this edge must be on $\mathcal{P}_2$. Let $p_j$ ($j \ge i$) be the next vertex from $\mathcal{P}$ that is also on path $\mathcal{P}_2$, such a vertex must exist as $p_k$ is on $\mathcal{P}$ and $\mathcal{P}_2$. By the construction of $G$, edge $(\overline{\ell}_j, p_j)$ (with weight $4kj$) must be in $\mathcal{P}_2$ which means the path $\mathcal{P}_2$ has weight at least $4k(k-i+1) + 4kj$ not including edges in the path from $\ell_i$ to $\overline{\ell}_j$. We also observe that any path from $\ell_i$ to $\overline{\ell}_j$ requires at least 4 edges, with total weight $4k$. If we have $j > i$, we immediately conclude that $\mathcal{P}_2$ has weight at least $4k(k-i+1 + j) + 4k \ge 4k(k+1) + 8k$. If $j=i$, then $\mathcal{P}_2$ contains a path from $\ell_i$ to $\overline{\ell}_i$. 
    This path can have length 4 if and only if the edges $(r_i, r_j')$ and $(\ell_j',\overline{\ell}_i)$ simultaneously exist for some $j$, which means  $S_a[(i-1) \cdot k+j] = S_b[(i-1)\cdot k +j] = 1$. This contradicts the assumption that strings $S_a$ and $S_b$ are disjoint. So, this $\ell_i$ to $\overline{\ell}_i$ path has length at least 8, which means $\mathcal{P}_2$ has weight at least $4k(k+1) + 8k = 4k^2+12k$. 
\end{proof}

To complete the reduction from set disjointness, assume that there is a CONGEST algorithm $\mathcal{A}$ that takes $R(n)$ rounds to compute the weight of a 2-SiSP path in a directed weighted graph on $n$ vertices. Consider the vertex partition $V_a, V_b$ of $V$ with $V_a = L \cup L' \cup \overline{L} \cup P$ and $V_b = R \cup R'$, and let $G_a(V_a, E_a), G_b(V_b, E_b)$ be the subgraphs of $G$ induced by the vertex sets $V_a, V_b$ respectively. Note that $G_a$ is completely determined by $S_a$ and $G_b$ is completely determined by $S_b$. Alice and Bob will communicate to simulate $\mathcal{A}$ on $G$. Alice will simulate the computation done in nodes in $V_a$, and Bob will simulate the computation done in nodes in $V_b$. If the algorithm communicates from a node in $V_a$ to a node in $V_b$, Alice sends all the information communicated along this edge to Bob. Since there are $2k$ cut edges, and $\mathcal{A}$ can send  $O(\log n)$ bits through each edge per round, Alice and Bob communicate up to $O(2k \cdot \log n)$ bits per round, for a total of $O(2k \cdot \log n \cdot R(n))$ bits. After the simulation, Alice knows $d_2(p_0, p_k)$ and can determine if the sets are disjoint by checking if $d_2(p_0, p_k) > 4k^2+9k-1$ (Lemma~\ref{lem:dirrplb}). Since any communication protocol for set disjointness must use at least $\Omega(k^2)$ bits and $n = \Theta(k)$, $R(n)$ is $\Omega\left(\frac{n}{\log n}\right)$. 

Our lower bound also applies to the RPaths problem, since given an algorithm $\mathcal{A}$ that computes replacement path for each edge, Alice can compute the minimum of those to get the second simple shortest path weight and then use Lemma~\ref{lem:dirrplb} as before. This lower bound applies even for graphs with constant undirected diameter: we can add a `sink' vertex with incoming edges from all vertices in $G$, so that Lemma~\ref{lem:dirrplb} still holds and the undirected diameter is 2.

\subsubsection{Directed Unweighted Replacement Paths Lower Bound}
\label{sec:dirunwrplb}

Our lower bound method uses a reduction from the undirected $s$-$t$ \textit{subgraph connectivity} problem defined in~\cite{sarma2012distributed} as follows: Given an undirected CONGEST network $G$ with $n$ vertices, a subgraph $H$ of $G$, and two vertices $s$, $t$, determine whether $s$ and $t$ are in the same connected component of $H$. The input subgraph $H$ is given by letting each vertex know which of its incident edges are in $H$. It is shown in~\cite{sarma2012distributed} that this problem has a lower bound of $\Omega\left(\frac{\sqrt{n}}{\log n} + D\right)$ in graphs with $D$ as small as $\Theta(\log n)$. We assume WLOG that network $G$ is connected.

\begin{figure}[t]
    \centering
    \includegraphics[scale=0.65]{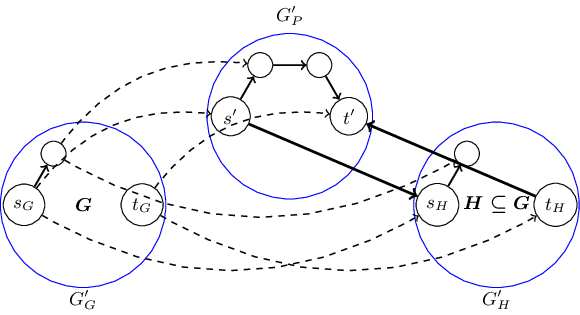}
    \caption{Directed unweighted RPaths, 2-SiSP lower bound graph $G'$}
    \label{fig:dirunwrplb} 
\end{figure}
\begin{proof}[Proof of Theorem~\ref{thm:undirrp}.\ref{thm:undirrp:lb} and Theorem~\ref{thm:dirrp}.\ref{thm:dirrp:approxlb}]
    Given an instance of $s$-$t$ subgraph connectivity with undirected network $G$, vertices $s$ and $t$ and subgraph $H$, our first attempt is to construct a directed unweighted graph $G'$ with two copies of $V(G)$: $G'_H$ contains only the edges in $H$, with bidirectional edges, and $G'_P$ contains only a directed shortest path from $s'$ to $t'$ made of edges in $G$ where $s',t'$ are copies of $s,t$. These copies are connected with directed edges  $(s',s_H)$ and $(t_H,t')$ (Figure~\ref{fig:dirunwrplb} without copy $G'_G$).

    This construction has the property that there is a second directed path from $s'$ to $t'$ in $G'$ (apart from the one in $G'_P$) if and only if there is an $s_H$-$t_H$ path in $G_H$. So 2-SiSP weight in $G$ is finite iff $s$,$t$ are connected in $H$. But, this construction could have high undirected diameter as we have no control over the diameter of $H$, and fails to give a meaningful lower bound.
    
    To obtain small undirected diameter, we add a third copy of $G$, denoted $G'_G$, which has all edges of $G$ as bidirectional edges, pictured in Figure~\ref{fig:dirunwrplb}. This copy is connected to the others with directed edges $(v_G, v_H)$ and $(v_G, v')$ where $v_G$ is the copy in $G'_G$ of $v \in G$. The undirected diameter of $G'$ is now $(D+2)$ ($D$ is the diameter of $G$) as we can connect any pair of vertices using a bidirectional path in $G'_G$ along with at most 2 connecting edges. This addition does not add any new directed paths from $s'$ to $t'$.

    Any communication in $G'$ can be simulated in a constant number of rounds in the underlying network of $G$, as each node $v$ in the network can simulate vertices $v_G,v_H,v'$ of $G'$, and all edges in $G'$ are either within the same node or have an underlying undirected edge of $G$. Constructing $G'$ requires only an $O(D)$-round computation of undirected shortest path from $s$ to $t$ in $G$. 

    This completes the reduction and establishes a lower bound of $\Omega\left(\frac{\sqrt{n}}{\log n} +D \right)$ for 2-SiSP (and RPaths) in unweighted directed graphs by additionally noting that $\Omega(D)$ rounds are necessary, as with other global problems in the distributed model~\cite{sarma2012distributed}, for information to travel to the farthest vertices to determine 2-SiSP. Our lower bound also applies to any $\alpha$-approximation algorithm ($\alpha > 1$) since we distinguish between 2-SiSP of length $\le n+2$ and infinite length. 
\end{proof}

\subsubsection{Other Directed Unweighted Graph Problems}
\label{app:dirunwlb}
\label{sec:dirunwother}

Our lower bound for directed unweighted RPaths can be adapted to give the same lower bound for other graph problems on directed unweighted graphs, including the basic problems of $s$-$t$ directed reachability and $s$-$t$ directed shortest path. These folklore lower bounds~\cite{ghaffari2015reach} have been attributed to~\cite{sarma2012distributed} which only deals with undirected graphs, and we make these results explicit. In Section~\ref{sec:dirmwclb} we will show an even stronger lower bound of $\Omega (n/\log n)$ for computing MWC and fixed-length cycle detection in directed unweighted graphs (also see Section~\ref{sec:cycledet}).

\begin{lemma}
    Any algorithm computing $s$-$t$ directed reachability or $s$-$t$ directed shortest path in a directed unweighted graph requires $\Omega\left(\frac{\sqrt{n}}{\log n} + D\right)$ rounds, even if the graph has undirected diameter as low as $\Theta(\log n)$. 
\end{lemma}
\begin{proof}
    We use a simpler version of the construction in Figure~\ref{fig:dirunwrplb} by removing $G'_P$ from the graph $G'$ to form a graph $G''$.  A directed path from $s_H$ to $t_H$ exists in $G''$ if and only if $s$ and $t$ are connected in the subgraph $H$. Using the same arguments as the RPaths lower bound, we note that $G''$  has undirected diameter $O(D)$ when the network $G$ has undirected diameter $D$ and $G''$ can be efficiently simulated on the original network $G$. So 
    we get the desired reduction for both problems from $s$-$t$ subgraph connectivity~\cite{sarma2012distributed}. 
\end{proof}

\subsubsection{Undirected Graphs RPaths Lower Bound}
\label{sec:undirrplb}
We show a lower bound of $\Omega\left(\frac{\sqrt{n}}{\log n} + D\right)$ for computing RPaths and 2-SiSP in undirected (weighted or unweighted) graphs, proving Theorem~\ref{thm:undirrp}.\ref{thm:undirrp:lb} using a reduction from $s$-$t$ weighted shortest path for which a $\Omega\left(\frac{\sqrt{n}}{\log n} + D\right)$ lower bound is known~\cite{elkin2006mst,sarma2012distributed}. Assume we are given an instance $G=(V,E)$ of $s$-$t$ weighted shortest path with $n$ vertices and undirected diameter $D$. We construct a graph $G'$ that is an undirected version of Figure~\ref{fig:dirunwrplb}, with two copies $G'_G$ and $G'_P$ (without the $G'_H$ copy) where $G'_G$ is just a copy of $G$ and $G'_P$ contains one undirected path from $s$ to $t$ using edges in $G$ with each edge weight being 1 (we assume network $G$ is connected). We connect these copies by adding undirected edges $(s_G,s')$ and $(t_G,t')$ with weight $n$. 
Following the arguments in Section~\ref{sec:dirunwrplb}, the new graph $G'$ has undirected diameter $\le D+2$. Constructing $G'$ requires only $O(D)$ rounds to compute an undirected $s$-$t$ path in $G$. We can simulate the graph $G'$ in the network $G$ by mapping vertices $v_G,v' \in V(G')$ to $v \in V$, then each edge in $G'$ either maps to an existing edge in $G$ or is within the same vertex of $G$, as in the argument used in Section~\ref{sec:dirunwrplb}.

The shortest path from $s'$ to $t'$ is the path in $G'_P$ (of weight at most $n$) and a second simple shortest path from $s'$ to $t'$ (with weight at least $2n$) is precisely a $s$-$t$ shortest path in $G$ along with the two connecting edges between $G'_G$ and $G'_P$. Thus, 2-SiSP distance in $G'$ is precisely $2n + d(s,t)$, where $d(s,t)$ is the $s$-$t$ shortest path distance in $G$, completing our reduction.

Hence the known lower bound of $\Omega\left(\frac{\sqrt{n}}{\log n} + D\right)$ for undirected weighted $s$-$t$ shortest path~\cite{sarma2012distributed} gives the same lower bound for undirected weighted 2-SiSP (and RPaths), even when parameters $h_{st}$ and $D$ are as low as $\Theta(\log n)$. In undirected unweighted graphs, we have an $\Omega(D)$ lower bound for RPaths and 2-SiSP for information to travel to the farthest vertices (similar to other global problems in the distributed model~\cite{sarma2012distributed}). This again matches the lower bound of $\Omega(D)$ for unweighted shortest path.

\subsection{Upper Bounds for Replacement Paths}

\subsubsection{Directed Weighted Replacement Paths}
\label{sec:dirrpub}

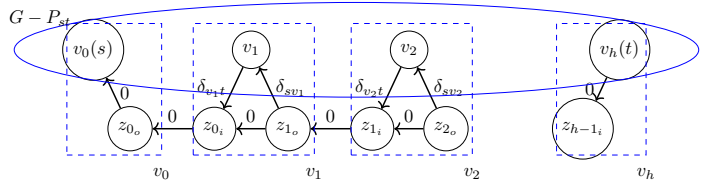
\begin{figure}
    \centering
    \tikzstyle{vertex}=[circle, draw=black,minimum size=20pt]
    \scalebox{0.7}{
    \begin{tikzpicture}
    
        \node[vertex] (v0) at (0,0) {$v_0(s)$};
        \node[vertex] (v1) at (3,0) {$v_1$};
        \node[vertex] (v2) at (6,0) {$v_2$};
        \node[vertex] (vk) at (10,0) {$v_h(t)$};

        \node[vertex] (z0o) at (0.7,-1.5) {$z_{0_o}$};
        \node[vertex] (z0i) at (2.3,-1.5) {$z_{0_i}$};
        \node[vertex] (z1o) at (3.7,-1.5) {$z_{1_o}$};
        \node[vertex] (z1i) at (5.3,-1.5) {$z_{1_i}$};
        \node[vertex] (z2o) at (6.7,-1.5) {$z_{2_o}$};
        \node[vertex] (zk1i) at (9.3,-1.5) {$z_{{h-1}_i}$};

        \path[draw,thick,->] (z0o) edge node[midway,right] {0} (v0);
        \path[draw,thick,->] (z0i) edge node[midway,above] {0} (z0o);
        \path[draw,thick,->] (v1) edge node[midway,left] {$\delta_{v_1 t}$} (z0i);
        \path[draw,thick,->] (z1o) edge node[midway,above] {0} (z0i);

        \path[draw,thick,->] (z1o) edge node[midway,right] {$\delta_{s v_1}$} (v1);
        \path[draw,thick,->] (z1i) edge node[midway,above] {0} (z1o);
        \path[draw,thick,->] (v2) edge node[midway,left] {$\delta_{v_2 t}$} (z1i);
        \path[draw,thick,->] (z2o) edge node[midway,above] {0} (z1i);

        \path[draw,thick,->] (z2o) edge node[midway,right] {$\delta_{s v_2}$} (v2);
        \path[draw,thick,->] (vk) edge node[midway,left] {0} (zk1i);

        \draw[blue] (5,0) ellipse (6.5cm and 0.9cm);
        \node (g) at (-1.0,1.0) [label=below:$G-P_{st}$] {};

        \draw[blue,dashed] (-0.5,0.5) rectangle (1.3,-2);
        \node (u) at (1.3,-2) [label=below:$v_0$] {};

        \draw[blue,dashed] (1.9,0.5) rectangle (4.2,-2);
        \node (u) at (4.2,-2) [label=below:$v_1$] {};

        \draw[blue,dashed] (4.9,0.5) rectangle (7.2,-2);
        \node (u) at (7.2,-2) [label=below:$v_2$] {};

        \draw[blue,dashed] (8.8,0.5) rectangle (10.5,-2);
        \node (u) at (10.5,-2) [label=below:$v_h$] {};

    \end{tikzpicture}
    }
    \caption{Directed weighted RPaths reduction to APSP}
    \label{fig:dirrpub} 
\end{figure}

In this section we present an $\tilde{O}(n)$ round CONGEST algorithm for computing RPaths and 2-SiSP in directed weighted graphs, which is nearly optimal given the near linear lower bound in Section~\ref{sec:dirrplb}. Our main tool is a reduction from RPaths to weighted APSP that can be simulated efficiently in the CONGEST network. This reduction is inspired by a sequential fine-grained reduction from RPaths to Eccentricities in~\cite{agarwal2018finegrained},
though some care is needed to ensure that the reduction can be efficiently mapped to the underlying CONGEST network. We present our algorithm to compute replacement path weights, and the second simple shortest path weight $d_2(s,t)$ can be computed by taking the minimum weight replacement path among those computed, with an additional $O(D)$ rounds.

Our algorithm constructs a graph $G'$ pictured in Figure~\ref{fig:dirrpub}, and runs a weighted APSP algorithm on $G'$. We show later how communication in the newly constructed $G'$ can be simulated efficiently in the underlying CONGEST network of $G$, so that the APSP algorithm can be applied to $G'$ in $\tilde{O}(n)$ rounds. The algorithm uses the $\tilde{O}(n)$ round weighted APSP algorithm~\cite{bernsteinapsp} as a subroutine and has $O(n)$ additive overhead, giving our $\tilde{O}(n)$ round bound.

We construct graph $G'(V',E')$ with $V' = V \cup Z_o \cup Z_i$, where $Z_o = \{z_{j_o} \mid 0 \le j < h\}$,  $Z_i = \{z_{j_i} \mid 0 \le j < h\}$. We denote the nodes on the shortest path $P_{st}$ by $s=v_0, v_1, \dots v_h=t$. $E'$ contains all edges in $E$ with their original weights, except the edges from the given shortest path $P_{st}$ which are removed. Additionally, $E'$ contains directed edges $(z_{j_o}, v_j), (z_{j_i}, z_{j_o}), (v_{j+1}, z_{j_i})$ for $0 \le j < h$. Edge $(z_{j_o}, v_j)$ has weight $\delta_{sv_j}$, edge $(v_{j+1}, z_{j_i})$ has weight $\delta_{v_{j+1}t}$ and edge $(z_{j_i}, z_{j_o})$ has weight 0 --- recall that $\delta_{sv_j}$ denotes the shortest path distance from $s$ to $v_j$ in $G$. We use $d'$ to denote shortest path distances in $G'$. The following lemma shows that we can compute replacement paths in $G$ using distances in $G'$.

\begin{lemma}
    \label{lem:dirrpub}
    The shortest path distance $d'(z_{j_o},z_{j_i})$ in $G'$~(Figure \ref{fig:dirrpub}) is equal to the replacement path weight $d(s,t,(v_j,v_{j+1}))$ in the original graph $G$.
\end{lemma}
\begin{proof}
    Let $\mathcal{P}$ be a replacement path for the edge $(v_j,v_{j+1})$ with weight $d(s,t,(v_j,v_{j+1}))$. We will construct a path from $z_{j_o}$ to $z_{j_i}$ that has the same weight as $\mathcal{P}$. Let $v_a$ be the first vertex where $\mathcal{P}$ deviates from $P_{st}$, $v_b$ be the first vertex after $v_a$ where $\mathcal{P}$ rejoins $P_{st}$. Note that $a \le j$ and $b \ge j+1$ as it is a replacement path for edge $(v_j,v_{j+1})$, and $\mathcal{P}$ contains a subpath $\mathcal{P}_{ab}$ from $v_a$ to $v_b$ that does not contain any edge from $P_{st}$.  Construct the path $\langle z_{j_o}, \dots z_{a_o}, v_a\rangle \circ \mathcal{P}_{ab} \circ \langle v_b, z_{{b-1}_i}, \dots z_{j_i}\rangle $, which has weight $w(z_{a_o}, v_a) + w(\mathcal{P}_{ab}) + w(v_b, z_{{b-1}_i}) = \delta_{sa} + w(\mathcal{P}_{ab}) + \delta_{bt} = w(\mathcal{P})$. Thus, we have $d'(z_{j_o},z_{j_i}) \le d(s,t,(v_j,v_{j+1}))$.
    
    Now, consider any shortest path $P$ from $z_{j_o}$ to $z_{j_i}$.
    Observe that any such $P$ must use a unique edge of the form $(z_{a_o}, v_a)$ and a unique edge $(z_{b}, z_{{b-1}_i})$ in order to reach $z_{j_i}$ from $z_{j_o}$, where $a \le j$ and $b \ge j+1$. Denote the subpath of $P$ from $v_a$ to $v_b$ as $P_{ab}$. Now, consider the path $\mathcal{P}$ in $G$ obtained by concatenating $s$-$v_a$ shortest path, $P_{ab}$ and $v_b$-$t$ shortest path --- this is a replacement path for edge $(v_j,v_{j+1})$ since $a \le j$, $b \ge j+1$ and $P_{ab}$ does not contain any edge on $P_{st}$. The weight of $\mathcal{P}$ is equal to $\delta_{sv_a} + w(P_{ab}) + \delta_{v_bt}$ which is equal to $w(P)$ since the only nonzero weight edges on $P$ that are outside subpath $P_{ab}$ are the ones with weight $\delta_{sv_a}$ and $\delta_{v_bt}$. Hence $d(s,t,(v_j,v_{j+1})) \le d'(z_{j_o},z_{j_i})$. 
\end{proof}

To simulate an APSP algorithm on $G'$ using the communication network $G$, we assign vertices $v_i, z_{{i-1}_i}, z_{i_o}$ of $G'$ to be simulated by CONGEST node $v_i$ of $G$ --- this is represented by the dashed boxes in Figure~\ref{fig:dirrpub}. This ensures that any edge of $G'$ corresponds to either a communication link between nodes in the CONGEST network of $G$, or the edge is within the same node of $G$. We can compute the weights required to simulate $G'$ after two SSSP computations with $s$ and $t$ as sources, and use an $\tilde{O}(n)$ algorithm to compute APSP in $G'$~\cite{bernsteinapsp}. We show how to augment this algorithm to construct replacement paths using routing tables in Section~\ref{sec:rpathsrecon}.

When the hop length of the $s$-$t$ path $h_{st}$ is small, the simple algorithm of performing $h_{st}$ shortest path computations with each edge on the $s$-$t$ path removed gives us an $O(h_{st} \cdot SSSP)$ round algorithm. We can obtain an improved round complexity if we only require a $(1 + \epsilon)$-approximation of the replacement path weight. We defer the presentation of this approximation algorithm for directed weighted RPaths to Section~\ref{sec:approxdirrpub} since it uses techniques that build on the directed unweighted RPaths algorithm.

\subsubsection{Directed Unweighted Replacement Paths}
\label{sec:dirunwrpub}

In this section, we present an algorithm for computing RPaths in directed unweighted graphs, proving Theorem~\ref{thm:dirunwrp}.\ref{thm:dirunwrp:ub}.
In the sequential setting, there are two approaches to compute replacement paths in directed graphs: (1) remove each edge in the input path $P_{st}$ and compute shortest paths in the resulting graphs separately, using $h_{st}$ shortest path computations~\cite{yen1971finding}, (2) compute shortest detour distances in order to compute replacement paths: A {\it detour} from $a$ to $b$, where $a,b$ are vertices on $P_{st}$, is a simple path from $a$ to $b$ with no edge in common with $P_{st}$. Any replacement path for edge $e\in P_{st}$ can be characterized as the concatenation of an initial $s$-$a$ subpath of $P_{st}$, a detour from $a$ to $b$, and a final $b$-$t$ subpath of $P_{st}$, where $a,b$ are vertices in $P_{st}$ such that $e$ is contained in the $a$-$b$ subpath of $P_{st}$~\cite{yen1971finding,roditty2012replacement}.

Our distributed algorithm uses both these approaches for different ranges of $h_{st},D$ (as in line~\ref{alg:dirunwrp:param} of Algorithm~\ref{alg:dirunwrp}). In the first method, used in Case 1 of Algorithm~\ref{alg:dirunwrp}, we compute replacement paths in $O(h_{st} \cdot SSSP)$ rounds using the obvious algorithm of removing one of the $h_{st}$ edges on the input shortest path $P_{st}$ and computing SSSP from $s$. We use a directed weighted SSSP algorithm with the weight of the removed edge set to $\infty$.  We do not use unweighted directed BFS since an $s$-$t$ shortest path could have up to $n-1$ hops after edge removal.

In the second method, used in Case 2 of Algorithm~\ref{alg:dirunwrp}, we present a distributed detour-based algorithm that runs in $\tilde{O}(n^{2/3}+\sqrt{nh_{st}}+D)$ rounds.

To compute short detours (hop length $\le h$, parameter $h$ determined in line~\ref{alg:dirunwrp:param}), our distributed algorithm exploits pipelining to compute $h$-hop limited BFS from each vertex on $P_{st}$ in $O(h_{st}+h)$ rounds~\cite{lenzen2019distributed,hoang2019round}. We compute these distances in the graph $G-P_{st}$, which is the graph $G$ with edges on $P_{st}$ removed. We denote shortest path distances in graph $G-P_{st}$ by $d^-(u,v)$.

For long detours (hop length $>h$), we sample $\Theta(p)$ vertices ($p$ determined in line~\ref{alg:dirunwrp:param}) in line~\ref{alg:dirunwrp:sample} and compute a `skeleton graph' on the set of sampled vertices: for $u,v \in S$, we add a directed edge $(u,v)$ to the skeleton graph with weight $d^-(u,v)$ if there is an $h$-hop directed shortest path from $u$ to $v$ in $G - P_{st}$. The edges of the skeleton graph are computed using an $h$-hop BFS in line~\ref{alg:dirunwrp:short}. The $h$-hop distances between all pairs of sampled vertices, and between sampled vertices and vertices on $P_{st}$ are broadcast to all vertices in line~\ref{alg:dirunwrp:broadcast}. Algorithm~\ref{alg:localdetour}, described below, is run at each vertex $a \in P_{st}$. It uses the $h$-hop distances (broadcast in line~\ref{alg:dirunwrp:broadcast} of Algorithm~\ref{alg:dirunwrp}) to {\it locally} compute at vertex $a$ all detours starting from $a$. It then computes at $a$ the best candidate replacement path among paths first deviating from $P_{st}$ at $a$ for each edge $e \in P_{st}$ that occurs after $a$ on $P_{st}$, denoted $d^a(s,t,e)$. Finally, Algorithm~\ref{alg:dirunwrp} performs a pipelined minimum operation along $P_{st}$ in line~\ref{alg:dirunwrp:globalrp} to compute shortest replacement path distances for all $h_{st}$ edges among candidate replacement paths computed by Algorithm~\ref{alg:localdetour} at each $a\in P_{st}$.

\begin{algorithm}[t]
    \caption{Directed Unweighted RPaths Algorithm}
    \begin{algorithmic}[1]
        \Require Graph $G=(V,E)$, vertices $s,t \in V$, $s$-$t$ shortest path $P_{st}$.
        \Ensure Replacement path distance $d(s,t,e)$ known at $s$ for each $e \in P_{st}$.
        \State \textbf{Case 1.} $D \le n^{1/4}, h_{st} \le n^{1/6}$ or $n^{1/4} < D \le n^{2/3}, h_{st} \le n^{1/3}$
        \State {Perform $h_{st}$ directed weighted SSSP computations in sequence, with each edge $e \in P_{st}$ having its weight set to $\infty$ to compute $d(s,t,e)$.}
        \State \textbf{Case 2.} $D \le n^{1/4}, h_{st}>n^{1/6}$ or $n^{1/4} < D \le n^{2/3}, h_{st} > n^{1/3}$ or $D>n^{2/3}$
        \State Fix parameter $p=n^{1/3}$ if $h_{st} < n^{1/3}$ and $p=\sqrt{n/h_{st}}$ if $h_{st} \ge n^{1/3}$, and fix $h = n/p$. \label{alg:dirunwrp:param}
        \State  Let graph $G - P_{st}$ be $G$ but with all edges in $P_{st}$ removed. \label{alg:dirunwrp:skeletonst}
        \State Sample each vertex $v \in G$ with probability $\Theta\left(\frac{\log n}{h}\right)$, let the set of sampled vertices be $S$. \label{alg:dirunwrp:sample} 
        \ForAll{vertex $v \in P_{st} \cup S$} 
        \mComment{Line~\ref{alg:dirunwrp:short} computes $h$-hop $S \times V(P_{st})$ distances and $h$-hop $S \times S$ distances (edges of the skeleton graph on $S$)}
        \State  {Perform BFS starting from $v$ on $G - P_{st}$, and the reversed graph, up to $h$ hops to compute unweighted shortest paths: for $u \in P_{st} \cup S$, both $v$ and $u$ know the $h$-hop limited distance $d^-(v,u)$. Since we have $(|S|+h_{st})$ sources, this takes $O(|S|+h_{st}+h)$ rounds.} \label{alg:dirunwrp:short}
        \State {Broadcast $\{d^-(v,u) \mid u \in S \text{ or } v \in S\}$. At most $(|S|^2 + h_{st}|S|)$ values are broadcast, taking $O(|S|^2+h_{st}|S|+D)$ rounds.}  \label{alg:dirunwrp:broadcast}
        \EndFor 
        \ForAll{vertex $a \in P_{st}$} 
            \mComment{Internally compute replacement paths $d^a(s,t,e)$ that deviate from $P_{st}$ at $a$, using distances broadcast in line~\ref{alg:dirunwrp:broadcast}}
            \State ComputeLocalRPaths($a$) \Comment{Algorithm~\ref{alg:localdetour}}
        \EndFor
        \For{edge $e \in P_{st}$} 
            \State {Compute $d(s,t,e) \gets \min_a d^a(s,t,e)$ by propagating values from $a \in P_{st}$ up the path $P_{st}$. The minimum for a single $e$ over all $a \in P_{st}$ takes $O(h_{st})$ rounds, and the computation for all $e \in P_{st}$ can be pipelined in $O(h_{st})$ rounds.} \label{alg:dirunwrp:globalrp}
        \EndFor \label{alg:dirunwrp:rpend}
    \end{algorithmic}
    \label{alg:dirunwrp}
\end{algorithm}

\begin{algorithm}[t]
    \caption{Local computation at $a \in P_{st}$ of candidate replacement paths deviating from vertex $a$}
    \begin{algorithmic}[1]
        \Require Graph $G=(V,E)$, shortest path $P_{st}$, subset $S \subseteq V$. The following distances in graph $G - P_{st}$ are known to $a$: $h$-hop distances $d^-(b,u)$, $d^-(u,b)$ for any $b \in P_{st}, u \in S$, $d^-(u,v)$, for $u,v \in S$, and  $d^-(a,b)$ for $b \in P_{st}$.
        \Ensure Vertex $a$ computes for each $e \in P_{st}$ after $a$ on $P_{st}$, the shortest replacement path distance $d^a(s,t,e)$ among paths that first deviate from $P_{st}$ at $a$.
        \Procedure{ComputeLocalRPaths}{$a$}
            \mComment{All computation is done internally using distances known to $a$.}
            \State Locally compute all pairs distances in skeleton graph on $S$: compute $d^-(y,z)$ for each $y,z \in S$ using the skeleton graph $h$-hop edge distances.\label{alg:dirunwrp:skeletondist}
            \ForAll{vertex $b \in P_{st}$ after $a$ along $P_{st}$} \label{alg:dirunwrp:detourst}
            \State Compute the best (short or long) detour $\delta(a,b)$  from $a$ to $b$ as \\ \;\; $\delta(a,b) = \min \left( d^-(a,b), \min_{u,v \in S} \left( d^-(a,u) + d^-(u,v) + d^-(v,b) \right) \right)$  \label{alg:dirunwrp:detourfinal}
            \EndFor \label{alg:dirunwrp:detourend}
            \For{edge $e = (x,y) \in P_{st}$ such that $a$ that appears before $x$ on $P_{st}$ or $a=x$} \label{alg:dirunwrp:rpst}
                \State {$d^a(s,t,e) = \min_{b \in P_{st}} \left( \delta_{sa} + \delta(a,b) + \delta_{bt}\right) $ (the minimum is over vertices $b$ that appear after $y$ on $P_{st}$ or $b=y$)} \label{alg:dirunwrp:localrp}
            \EndFor \label{alg:dirunwrp:localend}
        \EndProcedure
    \end{algorithmic}
    \label{alg:localdetour}
\end{algorithm}

\vspace{0.05in}
\noindent
{\bf Computation in Algorithm~\ref{alg:localdetour} (local computation at each $a\in P_{st}$).}
Algorithm~\ref{alg:localdetour} at vertex $a\in P_{st}$ takes as input the $h$-hop distances to and from $a$, computed in line~\ref{alg:dirunwrp:short} of Algorithm~\ref{alg:dirunwrp}, and the $h$-hop distances broadcast in line~\ref{alg:dirunwrp:broadcast} of Algorithm~\ref{alg:dirunwrp}. 
In Algorithm~\ref{alg:localdetour}, all pairs shortest path distances $d^-(u,v)$ for $u,v \in S$ in the skeleton graph are locally computed in line~\ref{alg:dirunwrp:skeletondist} using the $h$-hop skeleton graph edge distances. These distances, along with $h$-hop distances $d^-(u,b)$ for $u\in S, b \in P_{st}$, are used to compute long detours in line~\ref{alg:dirunwrp:detourfinal}. Short detours are computed at $a$ using the $h$-hop distance $d^-(a,b)$ to each vertex $b \in P_{st}$. With the best detour distances $\delta(a,b)$ computed in line~\ref{alg:dirunwrp:detourfinal}, $a$ locally computes replacement paths using detours starting from $a$ in line~\ref{alg:dirunwrp:localrp}, which gives the best candidate replacement path distance $d^a(s,t,e)$ among paths that first deviate from $P_{st}$ at $a$, for each edge $e$ after $a$ on $P_{st}$. 

\begin{lemma}
    The local computation in Algorithm~\ref{alg:localdetour} at $a \in P_{st}$ correctly computes $d^a(s,t,e)$, the minimum weight replacement path for $e \in P_{st}$ among paths that first deviate from $P_{st}$ at $a$. \label{lem:alglocal}
\end{lemma}
\begin{proof}
    We assume that $a$ knows the correct $h$-hop distances specified as input to Algorithm~\ref{alg:localdetour}. Note that the vertices and distances along $P_{st}$ are known to $a$ as part of RPaths input. 

    Due to our sampling probability, any shortest path between $u,v \in S$ can be decomposed into $h$-hop subpaths between sampled vertices w.h.p.\ in $n$. So in line~\ref{alg:dirunwrp:skeletondist}, vertex $a$ correctly computes all pairs shortest path distances between sampled vertices in $S$ using $h$-hop skeleton graph distances. 
    
    Now we show that line~\ref{alg:dirunwrp:detourfinal} computes a shortest detour $P^d_{ab}$ from $a$ to each $b \in P_{st}$ that occurs after $a$ on $P_{st}$, whose distance is denoted $\delta(a,b)$. If $P^d_{ab}$ is a short detour, with hop length $\le h$, its distance is equal to the $h$-hop distance $d^-(a,b)$ which is part of the input to $a$.  
    
    If $P^d_{ab}$ is a long detour, with hop length $> h$, we use the fact that due to our sampling probability, any path of $h$ hops contains a sampled vertex in $S$ w.h.p.\ in $n$. We can find a sampled vertex $u$ on the detour within $h$ hops from $a$ and a sampled vertex $v$ on the detour within $h$ hops from $b$. We will assume WLOG that $v$ occurs after $u$ or $u=v$. Then, the detour distance is $\delta(a,b) = d^-(a,u)+d^-(u,v)+d^-(v,b)$, and line~\ref{alg:dirunwrp:detourfinal} correctly computes this distance. 
    
    In any replacement path for edge $e \in P_{st}$ first deviating from $P_{st}$ at $a$, there is a vertex $b \in P_{st}$ where it rejoins $P_{st}$. We can characterize such a replacement path as the concatenation of the $s$-$a$ subpath of $P_{st}$, a detour $P^d_{ab}$ from $a$ to $b$, and the $b$-$t$ subpath of $P_{st}$. This path has weight $\delta_{sa} + \delta(a,b)+\delta_{bt}$. We then compute the minimum over all valid detour endpoints $b$ in line~\ref{alg:dirunwrp:localrp}. This correctly computes $d^a(s,t,e)$ for edges $e=(x,y)$ that are on the $a$-$b$ subpath of $P_{st}$ 
\end{proof}

\begin{lemma}
    Algorithm~\ref{alg:dirunwrp} computes replacement path weights in a directed unweighted graph in $\tilde{O}(\min(n^{2/3} + \sqrt{nh_{st}} + D, h_{st} \cdot SSSP))$ rounds. 
    \label{lem:dirumwrp}
\end{lemma}
\begin{proof}  
    We focus on the analysis of Case 2 since Case 1 is straightforward.

    \vspace{0in}
    \noindent
    \textbf{\textit{Correctness:}} The inputs used by Algorithm~\ref{alg:localdetour} at vertex $a \in P_{st}$ are correctly computed in Algorithm~\ref{alg:dirunwrp}: the $h$-hop distances from $a$ to other vertices $b\in P_{st}$ and $h$-hop distances from sampled vertices are computed in line~\ref{alg:dirunwrp:short}. After Algorithm~\ref{alg:localdetour} correctly computes $d^a(s,t,e)$, line~\ref{alg:dirunwrp:globalrp} computes $d(s,t,e) = \min_{a\in P_{st}} d^a(s,t,e)$ for each edge $e$ as the minimum distance among all valid replacement paths which may deviate at any $a \in P_{st}$.

    \vspace{0.0in}
    \noindent
    \textbf{\textit{Round Complexity:}} Recall that local computation does not contribute to the cost of an algorithm in the CONGEST model. So we can ignore Algorithm~\ref{alg:localdetour} for the round complexity analysis.
    In line~\ref{alg:dirunwrp:short} of Algorithm~\ref{alg:dirunwrp}, we use the $k$-source $h$-hop BFS algorithm for directed graphs which runs in $O(k+h)$ rounds using pipelining~\cite{lenzen2019distributed,hoang2019round}. As we have $k=p+h_{st}$ sources and $h$ hops, this takes $O(p+h_{st}+h)$ rounds. We use the standard broadcast operation in line~\ref{alg:dirunwrp:broadcast} which broadcasts $(p+h_{st})\cdot p$ values in $O(p^2+p\cdot h_{st} +D)$ rounds~\cite{peleg2000distributed}. The global minimum in line~\ref{alg:dirunwrp:globalrp} involves $h_{st}$ convergecast operations which can be pipelined to take $O(h_{st}+D)$ rounds. The total round complexity is $O(p^2+p\cdot h_{st} +h+D)$.

    Setting parameters $h = n^{2/3}, p = n^{1/3}$ gives us a round complexity of $\tilde{O}(n^{2/3} + n^{1/3} h_{st} + D)$. When $h_{st} \ge n^{1/3}$, the parameters $h = \sqrt{nh_{st}}, p = \sqrt{n/h_{st}}$ are more favorable, giving a round complexity of $\tilde{O}(\sqrt{nh_{st}} + n/h_{st} + D) = \tilde{O}(\sqrt{nh_{st}} + D)$ (since $h_{st} \ge n^{1/3}$). The input parameter $h_{st}$ can be shared to all nodes with a broadcast, so all vertices can choose the setting of $h, p$ appropriately. Thus, Case 2 takes $\tilde{O}(n^{2/3} + \sqrt{nh_{st}} + D)$ rounds.
\end{proof}

We augment Algorithm~\ref{alg:dirunwrp}, which computes only weights, to also construct replacement paths using routing tables within the same running time in Section~\ref{sec:dirunwrppath}.

\paragraph{Approximate Directed Weighted RPaths}
\label{sec:approxdirrpub}

We present a $(1+\epsilon)$-approximation algorithm for directed weighted RPaths that runs in $\tilde{O}\left(\sqrt{nh_{st}} + D + \min\left(n^{2/3},h_{st}^{2/5}n^{2/5+o(1)}D^{2/5}\right)\right)$ rounds. 

\begin{proof}[Proof of Theorem~\ref{thm:dirrp}.\ref{thm:dirrp:approxub}]
    Our algorithm is based on the directed unweighted RPaths algorithm described earlier. The key tool is to replace $h$-hop BFS computation in line~\ref{alg:dirunwrp:short} of Algorithm~\ref{alg:dirunwrp} with $(1+\epsilon)$-approximate $h$-hop limited shortest path computation, using an algorithm in~(\cite{nanongkai2014approx}, Theorem~3.6), which gives us a $\tilde{O}(k+h)$-round algorithm for $k$ sources.

    With this change, approximate distances are computed in the skeleton graph in line~\ref{alg:dirunwrp:skeletondist} of Algorithm~\ref{alg:localdetour}. Thus, the local detour distances (both short and long) are $(1+\epsilon)$-approximate detour distances in line~\ref{alg:dirunwrp:detourfinal}. The final replacement paths add these approximate detours to exact distances (line~\ref{alg:dirunwrp:globalrp} of Algorithm~\ref{alg:dirunwrp}) and are hence $(1+\epsilon)$-approximate. Using the same analysis as Lemma~\ref{lem:dirumwrp}, we get an algorithm with round complexity $\tilde{O}\left(n^{2/3} + \sqrt{nh_{st}} + D\right)$.

    When $h_{st}$ is small, we can improve the $h_{st}\cdot SSSP$ round algorithm used in the exact unweighted algorithm. A recent result~\cite{mwcarxiv} shows that $k$-source approximate directed weighted SSSP can be performed in $\tilde{O}(\sqrt{nk} + D)$ rounds if $k \ge n^{1/3}$ and in $\tilde{O}(\sqrt{nk} + D + k^{2/5}n^{2/5+o(1)}D^{2/5})$ rounds if $k<n^{1/3}$. We compute all detours using an $h_{st}$-source SSSP computation by treating each $a \in P_{st}$ as a source and computing shortest path distances in $G-P_{st}$. This method is efficient when $h_{st}< n^{1/3}$. Combining the two methods proves our result.
\end{proof}

\subsubsection{Undirected RPaths Algorithm}
\label{sec:undirrpub}
We will adapt the classical sequential 2-SiSP algorithm of~\cite{katoh1982efficient} to obtain our CONGEST algorithms for 2-SiSP and RPaths in undirected graphs.
In particular, the following result is shown in~\cite{katoh1982efficient}.

\begin{lemma}[Streamlined Undirected RPaths Characterization \cite{katoh1982efficient}]
    Given an undirected graph $G=(V,E)$ and vertices $s, t$, let $P_s(s,u)$, $P_t(u, t)$ denote the shortest path in $G$ from $s$ to $u$ and $u$ to $t$ respectively for any vertex $u$. Then, any replacement path of an edge in the shortest path $P_{st}$ is of the form $P_s(s,u) \circ (u,v) \circ P_t(v,t)$ (where $\circ$ denotes concatenation) for some vertices $u, v \in V$ such that $(u,v)$ is an edge.
    \label{lem:undirrp}
\end{lemma}

 We now prove Theorem~\ref{thm:undirrp}.\ref{thm:undirrp:ub} by describing and analyzing our CONGEST algorithms for undirected weighted RPaths with round complexity $O(SSSP + h_{st})$ and for undirected weighted 2-SiSP with round complexity $O(SSSP)$.

 \begin{proof}[Proof of Theorem~\ref{thm:undirrp}.\ref{thm:undirrp:ub}]
    The high-level idea of our algorithm is to compute shortest path trees rooted at $s$ and $t$, and compute the weights of each candidate replacement path of the form described in Lemma~\ref{lem:undirrp}. Then, we determine which edges each candidate path replaces, and compute the minimum candidate for each edge on $P_{st}$.

    We use a CONGEST SSSP algorithm~\cite{cao2023sssp} algorithm to compute shortest path distances $\delta_{su}$ and $\delta_{ut}$ for every vertex $u$, i.e., the weight of paths $P_s(s,u)$ and $P_t(u,t)$ respectively (we reverse the graph to compute the $\delta_{ut}$ values). 
    Each node $v$ sends the value $\delta_{vt}$ to its neighbors. With this information, $u$ computes the weight of concatenated path $P_s(s,u) \circ (u,v) \circ P_t(v,t)$ as $\delta_{su}+w(u,v)+\delta_{vt}$. Thus, $u$ internally computes the weights of these candidate replacement paths. We now determine the edges $e \in P_{st}$ for which these candidates are valid replacement paths.

    We track where the path $P_s(s,u)$ diverges from $P_{st}$ and where path $P_t(v,t)$ rejoins $P_{st}$, we define $\alpha(u) \in P_{st}$ as the last vertex on $P_s(s,u)$ that is also on the input path $P_{st}$. Similarly, we define $\beta(u) \in P_{st}$ to be the first vertex on $P_t(u,t)$ path that is also on $P_{st}$. Thus, the path $\mathcal{P}_{uv}$ is a replacement path for all edges on the $\alpha(u)$-$\beta(v)$ subpath of $P_{st}$.

    Since all nodes know the identities of vertices on $P_{st}$, we compute $\alpha$ (and similarly $\beta$) values during the SSSP computation by keeping track of when the shortest path first diverges from (or last rejoins) $P_{st}$. We will have node $u$ determine the set of edges on the $\alpha(u)$-$\beta(v)$ subpath for the candidate replacement path $\mathcal{P}_{uv}$ for each neighbor $v$, and to do this each node $v$ sends the value $\beta(v)$ to all its neighbors in one round. Note that the entire path $P_{st}$ is known to all vertices as part of the input, and hence this can be done locally at $u$. With this information, $u$ will now locally determine the best replacement path of the form $\mathcal{P}_{uv}$ for each edge on $P_{st}$.

    Finally, we perform a global minimum over all candidate replacement paths for each edge on $P_{st}$. For a single edge on $P_{st}$, this involves a convergecast~\cite{peleg2000distributed} operation where each node $u$ sends its locally determined best replacement path distance for the edge and a global minimum is computed among these distances. This takes $O(D)$ for a single edge, and the computation for all $h_{st}$ edges on $P_{st}$ can be pipelined to take $O(h_{st}+D)$ rounds. The total round complexity of our algorithm is $O(SSSP+h_{st})$ (the $D$ factor is subsumed by $SSSP$ complexity but $h_{st}$ could be large in a weighted graph).

    For 2-SiSP it suffices to perform one convergecast computation with the minimum value at each node $u$, and we do not need to pipeline $h_{st}$ different operations. This gives an algorithm computing 2-SiSP weight in $O(SSSP)$ rounds.
\end{proof}

We have shown that undirected weighted 2-SiSP has the same asymptotic complexity as undirected weighted SSSP in the CONGEST model, and the RPaths algorithm uses an additional $O(h_{st})$ rounds to compute the $h_{st}$ replacement paths. RPaths seems to require these additional $h_{st}$ rounds since we need to report on $h_{st}$ different replacement path distances, unlike 2-SiSP. Apart from this additive factor, our upper and lower bounds for undirected graphs match the current best bounds for SSSP in the CONGEST model. 

For unweighted graphs, the result in Lemma~\ref{lem:undirrp} is strengthened in~\cite{katoh1982efficient} to show that the replacement path is the concatenation of suitable $P_s(s,u)$ and $P_t(u,t)$ without the need of an intervening edge. However, we will simply use Lemma~\ref{lem:undirrp} even for the unweighted case so that we have similar algorithms for both cases. The undirected unweighted bound is $O(D)$ for both RPaths and 2-SiSP since $h_{st}$ is at most $D$, and this bound is optimal.

To our knowledge, the undirected unweighted variant of distributed RPaths is the only problem considered in this paper that has been previously studied in the literature: \cite{ghaffari2016fault}~studies the more general problem of computing single source replacement path distances, where given a vertex $s\in V$ we want to compute $d(s,t,e)$ for each $t \in V$ and $e \in E$. They are able to compute these distances at vertex $t$ in $\tilde{O}(D)$ rounds by a randomized scheduling of BFS computations. Their techniques do not immediately extend to weighted graphs, or directed graphs but is nearly optimal for undirected unweighted graphs.

%% file: section-mwc.tex
\section{Minimum Weight Cycle}
\label{sec:mwc}

\subsection{Lower Bounds for Minimum Weight Cycle}
\label{sec:mwclb}

\begin{figure}
    \centering
    \begin{minipage}{.4\columnwidth}
        \centering
        \tikzstyle{vertex}=[circle, draw=black,minimum size=20pt]
        \scalebox{0.7}{
        \begin{tikzpicture}
        
            \node[vertex] (l1) at (0,0) {$\ell_1$};
            \node[vertex] (l1p) at (2,0) {$\ell_1'$};
            \node[vertex] (r1) at (6,0) {$r_1$};
            \node[vertex] (r1p) at (8,0) {$r_1'$};
        
            \node[vertex] (l2) at (0,-1) {$\ell_2$};
            \node[vertex] (l2p) at (2,-1) {$\ell_2'$};
            \node[vertex] (r2) at (6,-1) {$r_2$};
            \node[vertex] (r2p) at (8,-1) {$r_2'$};
        
            \node[vertex] (lip) at (2,-2.5) {$\ell_i'$};
            \node[vertex] (rip) at (8,-2.5) {$r_i'$};
        
            \node[vertex] (lt) at (0,-4) {$\ell_k$};
            \node[vertex] (ltp) at (2,-4) {$\ell_k'$};
            \node[vertex] (rt) at (6,-4) {$r_k$};
            \node[vertex] (rtp) at (8,-4) {$r_k'$};
        
            \node (topmid) at (4,1.5) {};
            \node (botmid) at (4,-5) [label=left:$G_a$, label=right:$G_b$] {};
            
            \path[draw,thick,->] (l1) edge[bend left] (r1);
            \path[draw,thick,<-] (l1p) edge[bend left] (r1p);
        
            \path[draw,thick,->] (l2) edge[bend left] (r2);
            \path[draw,thick,<-] (l2p) edge[bend left] (r2p);
        
            \path[draw,thick,<-] (lip) edge[bend left] (rip);
            \path[draw,thick,<-] (l1) edge (lip) ;
            \path[draw,thick,->] (r1) edge (rip) ;
        
            \path[draw,thick,->] (lt) edge[bend left] (rt);
            \path[draw,thick,<-] (ltp) edge[bend left] (rtp);
        
            \path[draw=blue,dashed] (topmid) edge (botmid);
        \end{tikzpicture}
        }
        \caption{Directed unweighted MWC lower bound construction}
        \label{fig:dirmwc} 
    \end{minipage}
    \hspace{0.1\textwidth}%
    \begin{minipage}{.4\columnwidth}
        \centering
        \tikzstyle{vertex}=[circle, draw=black,minimum size=20pt]
        \scalebox{0.7}{
        \begin{tikzpicture}
        
            \node[vertex] (l1) at (0,0) {$\ell_1$};
            \node[vertex] (l1p) at (2,0) {$\ell_1'$};
            \node[vertex] (r1) at (6,0) {$r_1$};
            \node[vertex] (r1p) at (8,0) {$r_1'$};
        
            \node[vertex] (l2) at (0,-1) {$\ell_2$};
            \node[vertex] (l2p) at (2,-1) {$\ell_2'$};
            \node[vertex] (r2) at (6,-1) {$r_2$};
            \node[vertex] (r2p) at (8,-1) {$r_2'$};
        
            \node[vertex] (lip) at (2,-2.5) {$\ell_i'$};
            \node[vertex] (rip) at (8,-2.5) {$r_i'$};
        
            \node[vertex] (lt) at (0,-4) {$\ell_k$};
            \node[vertex] (ltp) at (2,-4) {$\ell_k'$};
            \node[vertex] (rt) at (6,-4) {$r_k$};
            \node[vertex] (rtp) at (8,-4) {$r_k'$};
        
            \node (topmid) at (4,1.5) {};
            \node (botmid) at (4,-5) [label=left:$G_a$, label=right:$G_b$] {};
            
            \path[draw,thick,-] (l1) edge[bend left] (r1);
            \path[draw,thick,-] (l1p) edge[bend left] (r1p);
        
            \path[draw,thick,-] (l2) edge[bend left] (r2);
            \path[draw,thick,-] (l2p) edge[bend left] (r2p);
        
            \path[draw,thick,-] (lip) edge[bend left] (rip);
            \path[draw,thick,dashed] (l1) edge node[midway,above] {2}  (lip) ;
            \path[draw,thick,dashed] (r1) edge node[midway,above] {2}  (rip) ;
        
            \path[draw,thick,-] (lt) edge[bend left] (rt);
            \path[draw,thick,-] (ltp) edge[bend left] (rtp);
        
            \path[draw=blue,dashed] (topmid) edge (botmid);
        \end{tikzpicture}
        }
        \caption{Undirected weighted MWC lower bound construction}
        \label{fig:undirmwc} 
    \end{minipage}%
\end{figure}

\subsubsection{Directed Minimum Weight Cycle}
\label{sec:dirmwclb}
We prove the near-linear lower bound for directed MWC and ANSC in Theorem~\ref{thm:dirmwc} using a reduction from Set Disjointness. We use the graph construction $G=(V,E)$ in Figure~\ref{fig:dirmwc}, and we use directed edges to ensure that the graph has short cycles if and only if the input sets are not disjoint. 

Consider an instance of the Set Disjointness problem where the players Alice and Bob are given $k^2$-bit strings $S_a$ and $S_b$ respectively. Construct $G$ with four sets of vertices: $L = \{\ell_i \mid 1 \le i \le k\}, L' = \{\ell_i' \mid 1 \le i \le k\}, R = \{r_i \mid 1 \le i \le k\}, R' = \{r_i' \mid 1 \le i \le k\}$. For each $1 \le i \le k$, add the directed edges $(\ell_i, r_i)$ and $(r_i', \ell_i')$. If $S_a[(i-1) \cdot k + j] = 1$, add the edge $(\ell_j', \ell_i)$ and if $S_b[(i-1) \cdot k + j] = 1$, add the edge $(r_i, r_j')$. We have directed the edges such that $\langle \ell_i, r_i, r_j', \ell_j' \rangle$ is a valid directed cycle when all the involved edges exist in $G$.

\begin{lemma}
    If $S_a \cap S_b \ne \phi$, $G$ has a directed cycle of length 4. If $S_a \cap S_b = \phi$, any directed cycle has length at least 8.
    \label{lem:dirmwclb}
\end{lemma}
\begin{proof}
    If the sets $S_a, S_b$ are not disjoint, then there exists $1\le i,j \le k$ such that $S_a[(i-1) \cdot k+j] = S_b[(i-1)\cdot k +j] = 1$, which means edges $(\ell_j', \ell_i)$ and $(r_i, r_j')$ exist in $G$. Hence, $\langle \ell_i, r_i, r_j', \ell_j' \rangle$ is a valid directed cycle of length 4 in $G$.

    Assume that the strings are disjoint. If we start our cycle (without loss of generality) with edge $(\ell_i, r_i)$, the next edge has to be $(r_i, r_j')$ for some $j$ such that $S_b[(i-1)\cdot k+j]=1$. From $r_j'$, we have only one outgoing edge to $\ell_j'$. From $\ell_j'$, we can take an edge $(\ell_j', \ell_p)$ for some $p$ such that $S_a[(p-1)\cdot k+j]=1$. Since the strings are disjoint and $S_b[(i-1)\cdot k+j]=1$, $p \ne i$. This means there is no cycle of length 4. It is easy to see that we can return to $\ell_i$ only after 4 more edges, which means any cycle in $G$ has length at least 8.
\end{proof}

\begin{proof}[Proof of Theorem \ref{thm:dirmwc}]
    Assume that there is a CONGEST algorithm $\mathcal{A}$ for directed unweighted MWC taking $R(n)$ rounds on a graph of $n$ vertices. Consider the partition$(V_a, V_b)$ of $G$ with $V_a$ containing vertex set $L \cup L'$ and $V_b$ containing vertex set $R \cup R'$. Alice simulates the subgraph $G_a$ induced by vertices in $V_a$ and Bob simulates the subgraph $G_b$ induced by vertices in $G_b$. Alice and Bob simulate $\mathcal{A}$ by communicating any information that $\mathcal{A}$ sends across edges in the cut separating $V_a, V_b$ and any communication within $G_a$ or $G_b$ are simulated by Alice and Bob internally. Since the cut has $4k$ edges and $\mathcal{A}$ is allowed to send up to $O(\log n)$ bits across an edge per round, simulating the computation of $\mathcal{A}$ on $G$ takes at most $O(4k \cdot \log n \cdot R(n))$ bits of communication between Alice and Bob. After simulating the MWC algorithm, we use Lemma~\ref{lem:dirmwclb} to determine if the sets are disjoint by checking if the resulting minimum length cycle has length at least 4. Since any communication protocol for Set Disjointness must use at least $\Omega(k^2)$ bits and $n = 4k$, we have shown that $R(n)$ is at least $\Omega\left(\frac{k^2}{4k \log n}\right) = \Omega\left(\frac{n}{\log n}\right)$. Clearly, this lower bound also applies to directed weighted graphs.
    
    Since $\mathcal{A}$ is able to distinguish between the presence of 4-cycles and cycles of length at least 8 in $G$, this lower bound also applies to $(2-\epsilon)$-approximation algorithms for computing MWC.
\end{proof}

\subsubsection{Undirected Weighted Minimum Weight Cycle}
\label{sec:undirmwclb}

We now prove the near-linear lower bound in Theorem~\ref{thm:undirmwc}.\ref{thm:undirmwc:lb} for undirected weighted MWC and ANSC. We use the graph construction $G=(V,E)$ in Figure~\ref{fig:undirmwc} and we use weights to ensure that the graph has short cycles if and only if the input sets are not disjoint. Consider an instance of the Set Disjointness problem where the players Alice and Bob are given $k^2$-bit strings $S_a$ and $S_b$. Construct $G=(V,E)$ with four sets of vertices: $L = \{\ell_i \mid 1 \le i \le k\}, L' = \{\ell_i' \mid 1 \le i \le k\}, R = \{r_i \mid 1 \le i \le k\}, R' = \{r_i' \mid 1 \le i \le k\}$. We add edges $(\ell_i, r_i)$ and $(\ell_i', r_i')$ for each $1 \le i \le k$. Each of these edges has weight 1. Add the edge $(\ell_i, \ell_j')$ if $S_a[(i-1) \cdot k + j] = 1$, and similarly add the edge $(r_i, r_j')$ if $S_b[(i-1) \cdot k + j] = 1$. These edges have weight 2.

\begin{lemma}
    If $S_a \cap S_b \ne \phi$, $G$ has a cycle of weight 6. If $S_a \cap S_b = \phi$, any cycle in $G$ has weight at least 8.
    \label{lem:undirmwclb}
\end{lemma}
\begin{proof}
    If the sets $S_a, S_b$ are not disjoint, then there exists $1\le i,j \le k$ such that $S_a[(i-1) \cdot k+j] = S_b[(i-1)\cdot k +j] = 1$, and hence the edges $(\ell_i, \ell_j')$ and $(r_i, r_j')$ exist in $G$. So, $G$ contains the cycle $\langle \ell_i, \ell_j', r_j', r_i, \ell_i \rangle$ which has weight 6.

    Now assume the sets are disjoint. First, notice that $G$ is bipartite and there is at most one edge of weight 1 incident to any vertex, which means any cycle in $G$ cannot have more 1-weight edges than 2-weight edges. So, any cycle of 6 or more edges has weight at least 9, and we can restrict our attention to cycles with 4 or fewer edges. Consider any cycle containing the edge $(l_i, r_i)$. It is easy to see that any such cycle must have at least 4 edges, and the only such cycle with 4 edges is of the form $\langle l_i, r_i, r_j', l_j' \rangle$ for some $j$. But, such a cycle exists in $G$ only if $S_a[(i-1) \cdot k + j] = 1$ and $S_b[(i-1) \cdot k + j] = 1$, which contradicts the assumption that $S_a \cap S_b = \phi$. The other type of cycle with 4 edges in this graph is of the form $ \langle l_i, l_j' , l_p , l_q' \rangle$, which has weight 8 since each edge in that cycle has weight 2. So, any cycle in $G$ has weight at least 8.
\end{proof}

\begin{proof}[Proof of Theorem~\ref{thm:undirmwc}.\ref{thm:undirmwc:lb}]
    To complete the reduction from Set Disjointness to MWC, assume that there is a CONGEST algorithm $\mathcal{A}$ for undirected MWC taking $R(n)$ rounds on a graph of $n$ vertices. 
    Consider the partition $(V_a, V_b)$ of $G$ with $V_a$ containing vertex set $L \cup L'$ and $V_b$ containing vertex set $R \cup R'$ and corresponding induced subgraphs $G_a, G_b$. Alice and Bob simulate $\mathcal{A}$ by communicating any information that $\mathcal{A}$ sends across edges in the cut separating $V_a, V_b$. Since the cut has $4k$ edges, simulating the computation of $\mathcal{A}$ on $G$ takes at most $O(4k \cdot \log n \cdot R(n))$ bits of communication between Alice and Bob. After simulating the MWC algorithm, we can use Lemma~\ref{lem:undirmwclb} to determine if the sets are disjoint by checking if the resulting minimum weight cycle has weight at least 6. Since any communication protocol must use at least $\Omega(k^2)$ bits and $n = 4k$, we have shown that $R(n)$ is $\Omega\left(\frac{n}{\log n}\right)$.

    Note that we can change the weight 2 on $(\ell_i, \ell_j')$ edges to any $w \ge 2$, which would mean our algorithm needs to distinguish between cycles of weight $2+2w$ and $4w$. Choosing a large enough $w$, we get the same hardness result for a $(2-\epsilon)$-approximation algorithm for undirected weighted MWC, for any constant $\epsilon > 0$.
\end{proof}

Lower bounds for computing undirected unweighted MWC has been studied previously with an $\Omega\left(\frac{\sqrt n}{\log n}\right)$ lower bound for exact computation of MWC that even applies to $(2-\epsilon)$-approximation~\cite{frischknecht2012, korhonen2017, drucker2014}.

\subsection{Upper Bounds for Minimum Weight Cycle}
\label{sec:mwcub}
\label{sec:undirmwcub}

\noindent
\textbf{Directed Graphs.}
As mentioned in the introduction, our algorithm for computing the MWC of directed graphs is straightforward: once we have computed all distances $\delta_{uv}$ using an APSP algorithm, computing $\min_{(u,v) \in E} \{\delta_{uv} + w(v,u)\}$ gives us the weight of the MWC of the graph. Computing this global minimum can be done in $O(D)$ rounds. So, we have an $O(APSP + D) = \tilde{O}(n)$-round algorithm for MWC on directed graphs, matching the lower bound up to logarithmic factors. We can also readily compute the smallest cycle passing through a node $v$ by computing $\min_{u: (u,v) \in E} \{\delta_{uv} + w(v,u)\}$, solving the ANSC problem in $O(APSP) = \tilde{O}(n)$ rounds.

\vspace{0.2cm}
\noindent
\textbf{Undirected Graphs.}
In the undirected case, computing MWC and ANSC needs a little more work since edges can be traversed in both directions, but we give near-linear algorithms for MWC and ANSC, stated in Theorem~\ref{thm:undirmwc}.\ref{thm:undirmwc:exact}, based on the characterization in Lemma \ref{lem:undir_cycle}.

\begin{lemma}
    Let $P_1 = \langle v_1, v_2, \dots v_i \rangle$ and $P_2 = \langle v_1, v_2', \dots v_j'\rangle$ be simple shortest paths between $v_1, v_i$ and $v_1,v_j'$ respectively such that $v_2 \ne v_2'$, $v_i \ne v_j'$ and $(v_i, v_j')$ is an edge in the graph. Then, $P_1 \cup P_2 \cup \{(v_i, v_j')\}$ contains a simple cycle passing through $v_1$ with weight at most $w(P_1) + w(P_2) + w(v_i, v_j')$. Further, a minimum weight cycle through $v_1$ will be of this form
    for some edge $(v_i,v_j')$.
    \label{lem:undir_cycle}
\end{lemma}
\begin{proof}
    Consider the walk obtained by traversing $P_1$, the edge $(v_i, v_j')$ and $P_2$ in reverse. This is a walk starting and ending in $v_1$, and contains the subwalk $(v_2, v_1, v_2')$ with $v_2 \ne v_2'$. So, the walk contains a simple cycle passing through $v_1$. Now, we prove that the minimum weight cycle through $v_1$ is of this form.

    Let $C = (v_1, v_2, \dots v_p)$ be the smallest cycle containing $v_1$. Let $d_C(v_1,v_i)$ denote the distance from $v_1$ to $v_i$ along the cycle $C$, and let $(v_i, v_{i+1})$ be the critical edge of $C$ with respect to $v_1$, i.e, it is the edge such that $\lceil w(C)/2 \rceil - w(v_i, v_{i+1}) \le d_C(v_1, v_i) \le \lfloor w(C)/2 \rfloor$ and $\lceil w(C)/2 \rceil - w(v_i, v_{i+1}) \le d_C(v_{i+1}, v_1) \le \lfloor w(C)/2 \rfloor$. We claim that $d_C(v_1, v_i)$ and $d_C(v_{i+1}, v_1)$ are shortest paths in $G$.

    Assume $d_C(v_1, v_i)$ is not a shortest path in $G$, and $P$ is a shortest path with weight $\delta_{v_1 v_i} < d_C(v_1, v_i)$. Let the vertex next to $v_1$ in $P$ be $v_k$. In the case that $v_k \ne v_p$, consider the walk consisting of $P$ along with the section of the cycle from $v_i$ to $v_1$. This is a walk starting and ending at $v_1$, and must contain a simple cycle containing $v_1$ --- this cycle has weight $< \delta_{v_1 v_i} + d_C(v_i,v_1) < w(C)$, which is a contradiction. If $v_k=v_p$, consider the walk consisting of $P$ along with the section of the cycle from $v_1$ to $v_i$ --- this walk contains a simple cycle of weight $< w(C)$ since the walk has weight $\delta_{v_1 v_i} + d_C(v_1, v_i) < 2\lfloor w(C)/2 \rfloor \le w(C)$ by definition of the critical edge. So, if we set $P_1$ to be the part of $C$ from $v_1$ to $v_i$ and $P_2$ to be the part from $v_{i+1}$ to $v_1$, then $P_1, P_2$ are shortest paths and the minimum weight cycle through $v_1$ has the form $P_1 \cup P_2 \cup \{(v_i, v_j')\}$.
\end{proof}

\begin{proof}[Proof of Theorem~\ref{thm:undirmwc}.\ref{thm:undirmwc:exact}]
    For our ANSC algorithm, we first compute APSP on the given graph with the modification that each vertex $v$ knows the shortest path distance $\delta_{uv}$ from every other vertex $u$, along with the first vertex $First(u,v)$ after $u$ in the path. We can readily modify the known APSP algorithms such as the one from~\cite{bernsteinapsp} to track this information. Then, each vertex sends the distance and first vertex information about all $n$ of its shortest paths to all its neighbors in $O(n)$ rounds. Now, each vertex $v$, for each neighbor $v' \in N(v)$, computes candidate cycles through some vertex $u$ consisting of paths from $(v,u)$ and $(u,v')$ along with the $(v,v')$ edge. To ensure this cycle contains a simple cycle passing through $u$, we check if $First(u,v) \ne First(u,v')$. 
    If this property holds we assign the cycle a weight of $\delta_{uv} + \delta_{uv'} + w(v,v')$ --- we call this a valid cycle passing through $u$.
    
    We compute the minimum valid cycle through $u$ by propagating the valid cycles at each node, taking the minimum at each step. This takes $O(n+D)$ rounds for all vertices $u$, giving a total round complexity of $O(APSP + n)$. For weighted graphs, this is an $\tilde{O}(n)$ round algorithm using the APSP algorithm of~\cite{bernsteinapsp}. For unweighted graphs, we can use an $O(n)$ round algorithm, e.g.,  \cite{lenzen2019distributed} to get an $O(n)$ round bound. Finally, for MWC we compute the global minimum of the ANSC values in $O(D)$ rounds.
\end{proof}

\subsection{Upper Bounds for Approximate Undirected Minimum Weight Cycle}

\subsubsection{Approximate Undirected Unweighted MWC}
\label{sec:approxundirmwcub}

\begin{algorithm}[t]
    \caption{2-Approximation Algorithm for Girth (Undirected Unweighted MWC Length)}
    \begin{algorithmic}
        \State 0: Construct set $S$ by sampling each vertex $v \in G$ with probability $\Theta(\log n/ \sqrt{n})$.
        \State 1: \textbf{for} {vertex $v \in G$}
        \State \;\; 1.A: Compute the $\sqrt{n}$-neighborhood for vertex $v$, i.e., the $\sqrt{n}$ closest vertices to $v$.
        \State \;\; 1.B: For each non-tree edge in the partial shortest path tree computed, record a candidate cycle.
        \State 2: \textbf{for} {vertex $w \in S$} 
        \State \;\; {2.A}: Perform BFS starting from $w$.
        \State \;\; {2.B}: For each non-tree edge in the BFS tree computed, record a candidate cycle.
        \State {3}: Return the minimum weight cycle among all recorded candidate cycles.
    \end{algorithmic}
    \label{alg:approxundirmwcub}
\end{algorithm}
    
We present an algorithm to compute an approximation of MWC length (or {\it girth}) in an undirected unweighted graph. A more detailed pseudocode and analysis of our algorithm can be found in Appendix A of our recent paper that focuses on MWC~\cite{mwcarxiv}. 

A CONGEST algorithm in~\cite{peleg2013girth} computes a $(2-\frac{1}{g})$-approximation of girth in $O(\sqrt{ng} + D)$ rounds, where $g$ is the girth. We significantly improve this result by removing the dependence on $g$ to give an $\tilde{O}(\sqrt{n} + D)$ round algorithm, proving Theorem~\ref{thm:undirmwc}.\ref{thm:undirmwc:approxunw}. We first present a 2-approximation algorithm, and then indicate how to improve it to a $(2-\frac{1}{g})$-approximation.

Our 2-approximation algorithm is given in Algorithm \ref{alg:approxundirmwcub}.
Algorithm \ref{alg:approxundirmwcub} samples $\tilde{O}(\sqrt{n})$ vertices and performs two computations. In Line \textbf{1}, the algorithm finds $\sqrt{n}$ closest vertices for each vertex $v$ in the graph and computes an upper bound on the minimum length of a cycle reachable from $v$ and contained entirely within this $\sqrt{n}$-neighborhood. In Line \textbf{2}, the algorithm performs a full BFS from the sampled vertices to compute 2-approximations of cycles `near' a sampled vertex. When a non-tree edge is encountered during the distance computations in lines~\textbf{1.A} and~\textbf{2.A}, the endpoints of the non-tree edge $(x,y)$ get two distance estimates $\delta_{wx}, \delta_{wy}$ from source $w$. Then, $\delta_{wx}+\delta_{wy}+1$ gives an upper bound on the length of a simple cycle containing $(x,y)$ --- these are the candidate cycles computed in lines~\textbf{1.B} and~\textbf{2.B}.

First, we argue the correctness of the algorithm. Consider the case when all vertices in a minimum weight cycle $C$ are contained in the $\sqrt{n}$-neighborhood of some vertex $v$ in $C$. When computing the $\sqrt{n}$-neighborhood of $v$, an edge of $C$ furthest from $v$ is a non-tree edge reachable from $v$ in its shortest path tree, and $C$ is computed as one of the candidate cycles in line~\textbf{1.B}. In this case, the MWC length is computed exactly. We prove a 2-approximation in the other case with the following lemma.

\begin{lemma}
    If no vertex $v$ in a minimum weight cycle $C$ contains all vertices of the cycle in its $\sqrt{n}$-neighborhood, lines \textbf{2.A} - \textbf{2.B} compute a 2-approximation of the length of $C$.
\end{lemma}
\begin{proof}
    Let $g$ be the length of $C$. Choose a vertex $v \in C$, and let its $\sqrt{n}$ neighborhood be $A$. Since $A$ does not contain all vertices in $C$, $A$ must have size $\sqrt{n}$ and hence w.h.p. $A$ contains some sampled vertex $w \in S$. Let $x$ be a vertex in $C$ but not in $A$, then we have $\delta_{vw} \le \delta_{vx}$ as $w$ is one of the $\sqrt{n}$ closest vertices to $v$ and $x$ is not. Also since $C$ has length $g$, $\delta_{vx} \le \lfloor g/2 \rfloor$ and hence $\delta_{vw} \le \lfloor g/2 \rfloor$.

    Now, line~\textbf{2.A} performs a BFS with $w$ as the source. There is some edge $(u,u')$ of $C$ that is a non-tree edge of the BFS tree rooted at $w$, and the candidate cycle recorded by this edge has length $\delta_{uw} + \delta_{u'w} + 1$. Using the triangle inequality $\delta_{uw} \le \delta_{uv} + \delta_{vw} = \delta_{uv} + \lfloor g/2 \rfloor$ (and similarly for $u'$), this candidate cycle has length at most $2\lfloor g/2 \rfloor + (\delta_{uv} + \delta_{u'v} + 1) \le 2g$. Note that $\delta_{uv} + \delta_{u'v} + 1 \le g$ since $u,u',v$ are all on $C$. So, line~\textbf{2.B} computes a 2-approximation of the MWC length.
\end{proof}

Now, we analyze the number of rounds. For line~\textbf{1.A}  we use an $(R+h)$-round algorithm in~\cite{lenzen2019distributed} to compute for each vertex its  $R$ closest neighbors within $h$ hops (the {\it source detection} problem); here we have $R = \sqrt{n}$ and $h=D$ so line~\textbf{1.A} runs in $O(\sqrt{n}+D)$ rounds.
 The set $S$ computed in line~\textbf{2} has size $\tilde{O}(\sqrt{n})$ w.h.p. in $n$  and the computation in line~\textbf{2.A} is an $|S|$-source BFS, which can be performed in $\tilde{O}(\sqrt{n}+D)$ rounds. The global minimum computation in line~\textbf{3} takes $O(D)$ rounds, giving a total round complexity of $\tilde{O}(\sqrt{n}+D)$.

We can improve this algorithm to a $(2-\frac{1}{g})$-approximation using a technique similar to~\cite{peleg2013girth}. For odd length MWC of length $g$, our algorithm already computes a cycle of length at most $(2g-1)$ in line~\textbf{2.B}. For even length MWC, we alter the source detection algorithm so that the cycle is identified in line~\textbf{1.B} even if exactly one vertex of the cycle is outside the $\sqrt{n}$-neighborhood --- this is readily done with one extra round since both parents of the vertex in the cycle are in the $\sqrt{n}$-neighborhood. With this change, an even cycle of length $g$ that is not detected in line~\textbf{1.B} gives a cycle of length at most $(2g-2)$ in line~\textbf{2.B}. So, we get an approximation ratio of $\frac{2g-1}{g} = \left(2 - \frac{1}{g}\right)$.

\subsubsection{Approximate Undirected Weighted MWC}
\label{sec:approxundirwtmwcub}

\begin{algorithm}[t]
    \caption{$(2+\epsilon')$-Approximation Algorithm for Undirected Weighted MWC}
    \begin{algorithmic}
        \State Let $h = \left(1+2/\epsilon\right) n^{3/4}$, and let $w^{i}(x,y) = \left\lceil \frac{2h w(x,y)}{\epsilon 2^i} \right\rceil$
        \State 0: Construct set $S$ by sampling each vertex $v \in G$ with probability $\Theta(\log n/ n^{3/4})$.
        \State 1: \textbf{for} $i=1$ \textbf{ to } $\log_{(1+\epsilon)} hW$
        \State \;\; {1.A}: Replace every edge $(x,y)$ with a path of length $w^{i}(x,y)$ to get graph $G^i$.
        \State \;\; {1.B}: Compute a $h$-hop limited $2$-approximation of the undirected unweighted MWC of $G^i$.
        \State \;\; {1.C}: If the cycle computed has hop length $H^i$, record a candidate cycle of weight $\frac{\epsilon2^i}{2h}H^i$ in $G$.
        \State 2: \textbf{for} {vertex $w \in S$} 
        \State \;\; {2.A}: Perform SSSP starting from $w$.
        \State \;\; {2.B}: For each non-tree edge in the shortest path tree computed, record a candidate cycle.
        \State {3}: Return the minimum weight cycle among all recorded candidate cycles.
    \end{algorithmic}
    \label{alg:approxundirwtmwcub}
\end{algorithm}

We give a $(2+\epsilon')$-approximation algorithm to compute undirected weighted minimum weight cycle in $\tilde{O}\left(\min(n^{3/4}D^{1/4} + n^{1/4}D, n^{3/4} + n^{0.65}D^{2/5} + n^{1/4}D, n)\right)$ rounds, proving Theorem~\ref{thm:undirmwc}.\ref{thm:undirmwc:approx}. This gives a sublinear round algorithm when $D$ is $\tilde{o}(n^{3/4})$ and compares favorably to the $\Omega(n)$ lower bound for $(2-\epsilon)$-approximation of undirected weighted MWC. For large $D$ ($D$ is $\tilde{\omega}(n^{3/4})$), the exact algorithm with $\tilde{O}(n)$ round complexity is more efficient.

In our recent paper on MWC~\cite{mwcarxiv}, we present an improved approximation algorithm that takes $\tilde{O}(n^{2/3}+D)$ rounds that is faster than the algorithm presented here, and is always sublinear when $D$ is sublinear. We leave this initial result here for completeness.

Our algorithm builds on the undirected unweighted algorithm in the previous section and uses scaling to approximate bounded hop-length cycles. This works well for cycles with short hop length, and we use sampling to handle cycles with long hop length. The scaling technique used here was used in the context of computing approximate shortest paths in~\cite{nanongkai2014approx}.

First, we slightly modify the undirected unweighted MWC approximation algorithm above to compute cycles of hop length up to $h$, for some parameter $h$. This is done by changing the computations in lines~\textbf{1.A} and~\textbf{2.A} of Algorithm~\ref{alg:approxundirmwcub} to terminate at hop length $h$. Then, the source detection takes $O(\sqrt{n}+h)$ and the BFS takes $\tilde{O}(\sqrt{n}+h)$ rounds, giving a total running time of $\tilde{O}(\sqrt{n}+h)$. This modified algorithm is used in line \textbf{1.B} of Algorithm~\ref{alg:approxundirwtmwcub}.

Lines~\textbf{1}-\textbf{1.C} compute a $(2+2\epsilon)$-approximation of the minimum weight cycle if it has hop length at most $n^{3/4}$ in $G$ --- we set $\epsilon' = 2\epsilon$ to get a $(2+\epsilon')$-approximation. To prove this, we use Lemma~3.4 of~\cite{nanongkai2014approx}, which says that in the scaled graph $G^i$, any distance $\delta^i_{xy}$ is such that $\frac{\epsilon2^i}{2h} \delta^i_{xy}$ is at most the distance $\delta_{xy}$ in the original graph $G$, and that there is an $i^*$ such that $\delta^{i^*}_{xy}$ is a $(1+\epsilon)$-approximation of $\delta_{xy}$. Further, for this $i^*$, $\delta^{i^*}_{xy} \le h \left(1+2/\epsilon\right) n^{3/4} $ if the $x$-$y$ path has hop length at most $n^{3/4}$. These arguments about distances also apply to cycles, so if $G$ has a cycle $C$ of hop length at most $n^{3/4}$ then we detect a $2(1+\epsilon)$-approximation of it in line~\textbf{1.C}. If the MWC has hop length more than $n^{3/4}$, lines~\textbf{2}-\textbf{2.B} exactly compute the MWC. The sampled set $S$ has size $\tilde{O}(n^{1/4})$ and w.h.p. hits every cycle with at least $n^{3/4}$ vertices. So, the MWC contains a sampled vertex $w$, and the SSSP from $w$ in line~\textbf{2.A} detects the cycle.

Lines~\textbf{1.A}-\textbf{1.C} take $\tilde{O}(\sqrt{n}+h)$ rounds, and performing the computation for all $i$ takes $\tilde{O}(n^{3/4} \cdot \log W)$. We perform an SSSP from $\tilde{O}(n^{1/4})$ vertices in lines~\textbf{2.A} and we use the $\tilde{O}(\sqrt{n}D^{1/4}+D)$-round SSSP algorithm of~\cite{chechiksssp}, for a total of $\tilde{O}(n^{3/4}D^{1/4} + n^{1/4}D)$ rounds. The total round complexity is $\tilde{O}(n^{3/4}D^{1/4} + n^{1/4}D)$. 

We can also use a $(1+\epsilon)$-approximation algorithm for SSSP in line~\textbf{2.A} to compute a $(2+\epsilon')$-approximation of the MWC. Each approximate SSSP computation takes $\tilde{O}\left( \sqrt{n} + n^{2/5}D^{2/5} + D\right)$ rounds using the algorithm of~\cite{cao2021approximatesssp}, which is more efficient than exact SSSP computation when $D$ is $\tilde{o}(n^{2/3})$. This gives an improved algorithm when $D$ is $\tilde{o}(n^{2/3})$ with round complexity $\tilde{O}\left(n^{3/4} + n^{0.65}D^{2/5} + n^{1/4}D\right)$.

\subsection{Fixed-Length Cycle detection}
\label{sec:cycledet}
The problem of detecting small cycles of fixed length in undirected unweighted graphs has been studied extensively in the CONGEST model. Triangles (3-cycles) can be detected in $\tilde{O}(n^{1/3})$ rounds~\cite{chang2021near,izumi2017triangle}, and there is a matching lower bound~\cite{izumi2017triangle} for the enumeration of all triangles in a graph. For detection of 4-cycles $C_4$ in undirected graphs, a nearly tight bound of $\tilde{\Theta}(n^{1/2})$ is given in~\cite{drucker2014}, and a nearly tight bound of $\tilde{\Theta}(n)$ for the detection of $C_q$ in undirected graphs for any odd $q \ge 5$. A lower bound of $\tilde{\Omega}(n^{1/2})$ was given in~\cite{korhonen2017} for the detection of $C_{2q}$ for any $q \ge 3$. We note here that the $\tilde{\Omega}(n)$ lower bound for odd cycles given by~\cite{drucker2014} cannot be directly extended to a lower bound of undirected unweighted MWC since the lower bound graph in their construction necessarily contains short $C_4$ cycles. 
An $\tilde{\Omega}(n)$ lower bound for the problem of enumerating all $C_4$-copies in a given graph
is given in~\cite{eden2021sublinear}. However, their construction does not extend to the problem of detecting if a copy of $C_4$ exists in a given graph. 

Turning to directed graphs, the algorithm in~\cite{chang2021near} for 3-cycle detection continues to apply in the directed case~\cite{Pettie2022}, and this gives an $\tilde{O}(n^{1/3})$ round algorithm for detecting directed 3-cycles. Surprisingly, for cycles of length $q \ge 4$, we have an $\tilde{\Omega}(n)$ lower bound for detecting directed cycles of length $q$ as stated in Theorem~\ref{thm:dirunwother}.\ref{thm:dirunwother:cycledet}, which we now establish.

\begin{proof}[Proof of Theorem~\ref{thm:dirunwother}.\ref{thm:dirunwother:cycledet}]
We build on the reduction used in the directed MWC lower bound. Replace each $\ell_i$ vertex in Figure~\ref{fig:dirmwc} with a directed path of $q-3$ vertices. Add incoming edges to $\ell_i$ to the first vertex of this path and outgoing edges from $\ell_i$ to the last vertex of this path. Clearly, a 4-cycle involving $\ell_i$ in the original graph is a $q$-cycle in this modified graph. Using the same argument as in Lemma~\ref{lem:dirmwclb}, we show that Set Disjointness reduces to checking whether this graph has a $q$-cycle or if the smallest cycle in the graph has length $2q$. This graph has $n=(q-3)k+3k = \Theta(k)$ vertices and we get $R(n) \cdot 2k \cdot \log n \ge \Omega(k^2)$ which shows that $R(n)$ is $\Omega(\frac{n}{\log n})$.
\end{proof}

%% file: section-path-constr.tex
\section{Path and Cycle Construction}
\label{sec:recon}

We have presented replacement path algorithms that compute the weights of the replacement paths for each edge in the input shortest path $P_{st}$. In Section~\ref{sec:rpathsrecon}, we extend the algorithms to construct the actual replacement path in a distributed manner when an edge on $P_{st}$ fails. Similarly for minimum weight cycle, in Section~\ref{sec:mwcrecon} we extend algorithms that compute weights to construct a minimum weight cycle.

\subsection{Replacement Path Construction}
\label{sec:rpathsrecon}

In this section, we extend the algorithms that compute replacement path weights (from Section~\ref{sec:rp}) to construct the actual replacement path in a distributed manner when an edge on $P_{st}$ fails. For this purpose, we construct a routing table at each node. This routing table has size $h_{st}$ ($h_{st}$ is the hop length of $P_{st}$), and will allow us to efficiently re-establish communication on failure.

For undirected graphs, we additionally show an alternate on-the-fly model of path construction in Section~\ref{sec:undirpathconstr}, where we will construct replacement paths by storing only $O(1)$ information per node, without storing the whole routing table, at the cost of slightly increasing the rounds required to construct paths. On-the-fly construction of replacement path can be done with only a constant factor more rounds than with routing table (as shown later).

\vspace{0.1in}
\noindent
\textbf{Routing Table Description}: Each node $v$ has a routing table $R_v$ of size $h_{st}$. For each edge $e \in P_{st}$, there is a routing table entry $R_v(e) = u$, where $u$ is the next vertex on the replacement path for $e$ if $v$ is on the path.

\vspace{0.1in}
\noindent
\textbf{Path Construction from Routing Table}: Once routing tables are computed in a preprocessing step we must establish communication along a replacement path when an edge failure occurs. Say an edge $e$ on $P_{st}$ fails, and a failure message is broadcast from a node incident to $e$. This message reaches $s$ in at most $h_{st}$ rounds as any edge on $P_{st}$ is at most $h_{st}$ hops from $s$. $s$ is the first vertex on the replacement path for $e$, and checks the routing table entry $R_s(e) = v_1$ to get the next vertex. $s$ sends a message notifying $v_1$, and $v_1$ checks its routing table entry $R_{v_1}(e) = v_2$. This procedure continues until vertex $t$ is reached, and the replacement path $\langle s,v_1,v_2, \dots v_k=t \rangle$ has been established. This process takes $k = h_{rep}$ rounds, where $h_{rep}$ is the hop length of the replacement path. Thus, after an edge fails, we use a total of $h_{st}+h_{rep}$ rounds to re-establish communication.

\subsubsection{Directed Weighted Replacement Path Construction}

We will now construct routing tables with only $O(n)$ overhead, thus the total round complexity of the algorithm computing weights and constructing paths is unchanged at $\tilde{O}(n)$.

\vspace{0.1in}
\noindent
\textbf{Routing Table Construction for Directed Weighted RPaths}:  We modify the APSP algorithm such that each vertex $u$ also knows the vertex $First(u,v)$ which is the first vertex after $u$ on the $u$-$v$ shortest path in $G'$ for each $v$, and $v$ knows $Last(u,v)$ which is the last vertex before $v$ on the $u$-$v$ shortest path. Consider an edge $e = (v_j,v_{j+1})$ on $P_{st}$, we prove in Section~\ref{sec:dirrpub} that the $z_{j_o}$-$z_{j_i}$ shortest path has the same weight as the replacement path for $e$. We further observe that if $v_a$, $v_b$ are the first and last vertices respectively on $P_{st}$ that are also on this $z_{j_o}$-$z_{j_i}$ shortest path, then the replacement path for $e$ is obtained by concatenating the $s$-$v_a$ subpath of $P_{st}$, $v_a$-$v_b$ shortest path in $G'$ and $v_b$-$t$ subpath of $P_{st}$. 

Starting from $z_{j_o}$, we follow the next vertex ($First(z_{j_o},z_{j_i})$) on the $z_{j_o}$-$z_{j_i}$ shortest path until we reach a vertex $v_a$ that is also on $P_{st}$. Similarly, we start from $z_{j_i}$ and follow $Last(z_{j_o},z_{j_i})$ until we reach vertex $v_b$ on $P_{st}$. Once they are found, the identities of $v_a, v_b$ are broadcast to all vertices. We may need to follow up to $O(n)$ vertices to find $v_a, v_b$, and the computation for all edges in $P_{st}$ can be pipelined since each edge transmits only up to 2 messages at any round (one each for $First$ and $Last$ traversal). Thus, we take $O(n)$ rounds to find $v_a,v_b$ and $O(h_{st}+D) = O(n)$ rounds to broadcast $v_a,v_b$ for all edges in $P_{st}$.

After the broadcast, we now set routing table entries for $e\in P_{st}$. Each vertex $u$ on the $s$-$v_a$ subpath of $P_{st}$ sets $R_u(e) = u'$ where $u'$ is the next vertex on the subpath. Since all vertices know the identities of vertices on $P_{st}$, this can be done locally at $u$ once $v_a$ is known. Similarly vertices on $v_b$-$t$ subpath set $R_u(e)=u'$ appropriately. All other vertices set $R_u(e) = First(u, v_b)$, so that once vertex $v_a$ is reached by the routing algorithm, the $v_a$-$v_b$ shortest path in $G'$ is traversed.

\begin{theorem}
    \label{thm:dirrprecon}
    The directed weighted RPaths algorithm can be modified to compute routing tables for replacement path construction in $\tilde{O}(n)$ rounds. Storing the routing table uses $O(h_{st})$ space at each node.
    After edge failure, path construction takes $h_{st}+h_{rep}$ rounds.
\end{theorem}

\subsubsection{Directed Unweighted Replacement Path Construction}
\label{sec:dirunwrppath}

We now show how to construct routing tables for directed unweighted RPaths. The round complexity of constructing routing tables in addition to computing weights is unchanged, as we show the overhead for routing table construction is low ($O(h_{st}+D+h)$ in terms of parameter $h$ of Algorithm~\ref{alg:dirunwrp}).

\vspace{0.1in}
\noindent
\textbf{Routing Table Construction in Directed Unweighted Graphs}:
In Case~1 of Algorithm~\ref{alg:dirunwrp}, this is done simply by tracking the next vertex $First(u,v)$ in a shortest path computed by the SSSP algorithm in the graph with each edge $e \in P_{st}$ removed, and setting the routing table entry for all $v\in V$ as $R_v(e) = First(v,t)$ to route along the $s$-$t$ shortest path. 

For Case 2, we now show how to construct routing tables without much overhead. We first modify the BFS in line~\ref{alg:dirunwrp:short} such that $u$ knows the parent vertex in the computed $u$-$v$ shortest path in $G-P_{st}$, denoted $First(u,v)$. Consider an edge $e \in P_{st}$, the replacement path for $e$ consists of subpaths of $P_{st}$ from $s$ to $a$ and $b$ to $t$ along with a detour from $a$ to $b$ for some $a,b \in P_{st}$. During the minimum detour computations in line~\ref{alg:dirunwrp:localrp} of Algorithm~\ref{alg:localdetour} and line~\ref{alg:dirunwrp:broadcast} of Algorithm~\ref{alg:dirunwrp}, we keep track of the vertices $b$ and $a$ respectively for which the replacement path distance is minimum. Now, the detour vertices $a$ and $b$ are known for each of the $h_{st}$ replacement paths, and this information is broadcast to all vertices (which takes $O(h_{st}+D)$ rounds). 

Each vertex $v$ on the $s$-$a$ subpath or $b$-$t$ subpath of $P_{st}$ sets $R_v(e) = v'$ where $v'$ is the next vertex on the subpath. This can be done locally at $v$ as the identities of vertices on $P_{st}$ are known to all vertices. Now, we set routing tables for all other vertices to route along the detour from $a$ to $b$. If $D(a,b)$ is a short detour, set $R_v(e) = First(v,b)$ for each $v$ on the detour starting from $a$  in $O(h)$ rounds. If $D(a,b)$ is a long detour, after line~\ref{alg:dirunwrp:detourfinal} of Algorithm~\ref{alg:localdetour}, $a$ knows that the detour is a concatenation of shortest paths $a$-$v$, $v$-$u$, $u$-$b$ in $G-P_{st}$, where $u,v$ are sampled vertices and $a$-$v$,$u$-$b$ are $h$-hop paths. The identities of $u$,$v$ are broadcast. We set routing table entries at vertex $x$ along the $a$-$u$ and $v$-$b$ paths by traversing the paths in $O(h)$ rounds and setting $R_x(e)=First(x,u)$ and $R_x(e)=First(x,b)$ respectively. Since each $v$-$u$ shortest path can be locally computed at each vertex $x$ (line~\ref{alg:dirunwrp:skeletondist} of Algorithm~\ref{alg:localdetour}), each vertex on the $v$-$u$ path can locally set $R_{x}(e)=First(x,v')$ where $v'$ is the next sampled vertex after $x$ on the $v$-$u$ shortest path. So, the overhead for routing table construction is $O(h)$ for traversing short hop paths, and $O(h_{st}+D)$ for the broadcasts. Since this overhead is less than the round complexity of Algorithm~\ref{alg:dirunwrp} (which is $O(p^2+p\cdot h_{st} +h+D)$ in terms of parameters $p,h$), the total round complexity of constructing routing tables in addition to computing weights is unchanged.

\begin{theorem}
    \label{thm:dirunwrprecon}
    The directed unweighted RPaths algorithm can be modified to also compute routing tables for replacement path construction with the same round complexity of $\tilde{O}(n^{2/3}+\sqrt{nh_{st}}+h_{st}, \min(h_{st} \cdot SSSP))$ rounds. Storing the routing table uses $O(h_{st})$ space at each node. After edge failure, path construction takes $h_{st}+h_{rep}$ rounds.
\end{theorem}

\subsubsection{Undirected Replacement Path Construction}
\label{sec:undirpathconstr}

As noted in~\cite{bremler2001restoration}, undirected replacement paths can be restored using an arbitrary tie-breaking scheme to get a unique shortest path tree when using the characterization of replacement paths as two shortest paths along with an edge. There has been recent progress with a similar result~\cite{bodwin2023restorable} for unweighted undirected graphs where replacement paths are characterized by the concatenation of just two shortest paths, but we do not use this tie-breaking scheme for our application.

We will construct undirected replacement paths using routing tables. Additionally, we propose an on-the-fly path construction model for undirected replacement paths.

For undirected weighted graphs, routing table construction takes $\tilde{O}(h_{st}+h_{rep})$ rounds, which may exceed the $O(h_{st} + SSSP)$ round complexity of the algorithm that just computes weights. This is in contrast to the directed case, and the undirected unweighted case, where routing table construction rounds is within the rounds used to compute weights.

\vspace{0.1in}
\noindent
\textbf{Routing Table Construction for Undirected RPaths}: Consider an edge $e \in P_{st}$, we track the deviating edge $(u,v)$ for each replacement path computed by our algorithm, when finding the candidate replacement path for $e$ with minimum weight through the convergecast procedure. We broadcast this deviating edge information for each of the $h_{st}$ edges in $O(h_{st}+D)$ rounds. Now, we show how to set routing table entries to construct the replacement path $P_s(s,u) \circ (u,v) \circ P_t(v,t)$ for edge $e$. 

During the SSSP computation from $t$, for each vertex $x$ we track the next vertex on the $x$-$t$ path, denoted $First(x,t)$. For each $x$ we initially set $R_x(e)=First(x,t)$ to route along $P_t(v,t)$. On the other hand, the next vertex information for $P_s(s,u)$ cannot be trivially determined during the SSSP from $s$, since only $u$ knows its parent on the $u$-$s$ path and $s$ does not know the next vertex on the $s$-$u$ path without additional computation. 

We start from $u$, and $u$ informs its parent $u'$ that it is the next vertex on the $P_s(s,u)$ path, and we set $R_{u'}(e) = u$. In the next round, $u'$ informs its parent on the $P_s(s,u')$ path, and so on until $s$ gets the message and sets its routing table entry. Repeating this operation up the path to $s$ takes $O(h_{rep})$ rounds and correctly sets routing table entries for nodes on path $P_s(s,u)$. Finally, we set $R_u(e)=v$ to route along edge $(u,v)$. Computing routing tables for a single edge takes $h_{rep}$ rounds, and we can use random scheduling~\cite{ghaffarischeduling} to perform the computation for all $h_{st}$ edges since each computation involves only one message per edge. The total congestion along any edge is $O(h_{st})$, and the computation takes $\tilde{O}(h_{st}+h_{rep})$ rounds. In contrast to the earlier construction algorithms for directed graphs, this may be more than the $O(h_{st} + SSSP)$ round complexity to compute just the weights of replacement paths if $h_{rep}$ is large.

For the unweighted case, we use the same algorithm but $h_{rep} \le 2D+1$ since both $s$-$u$ and $u$-$t$ shortest paths have at most $D$ hops for any vertex $u$. So, we need $O(D)$ rounds to compute routing tables, which is the same round complexity needed to compute replacement path weights.

\vspace{0.1in}
\noindent
\textbf{On-the-fly Model}: In this model, once the edge that has failed is known, this information is broadcast and the nodes should determine a replacement path for this failed edge without storing the entire routing table. For undirected graphs, we are able to use only $O(1)$ words of storage per node while establishing communication along a replacement path in $h_{st} + 3h_{rep}$ rounds (which is only a constant factor higher than routing tables). The advantage here is that it may not be necessary to store $h_{st}$ words of information for the entire routing table at every node, and this model is preferable if space is a concern.

\vspace{0.1in}
\noindent
\textbf{On-the-fly Construction for Undirected RPaths}:
In the on-the-fly model, we store information about the deviating edge $(u,v)$ at $u$, and each node $x$ stores $First(x,t)$. Once the failing edge $e$ is known, node $s$ is notified of the failure by the vertex incident to $e$ in at most $h_{st}$ rounds and $s$ sends this information down its shortest path tree until it reaches $u$ in $h_{rep}$ rounds. The deviating edge $(u,v)$ corresponding to $e$ is now determined, and the routing information along $P_s(s,u)$ is computed starting from the parent of $u$ through the shortest path to $s$ as described in the routing table construction procedure in $h_{rep}$ rounds. Node $u$ sets its next vertex to $v$ and all other nodes $x$ set their next vertex to be $First(x,t)$ to route along $P_t(v,t)$. Now every vertex on the replacement path knows its next vertex on the path and messages can be sent in $h_{rep}$ rounds, establishing communication along the replacement path in $h_{st}+3h_{rep}$ rounds. Construction in the routing table model takes $h_{st} + h_{rep}$ rounds, so the on-the-fly construction uses only a constant factor more rounds than the routing table model with $O(1)$ space per node. 

For directed replacement paths, it seems difficult to have a similar on-the-fly model that uses $o(h_{st})$ space per node. In the directed case, a single vertex $v_a$ on the $P_{st}$ path may be the first deviating vertex of $\Omega(h_{st})$ replacement paths, and each of the corresponding detours may be very different. It is not clear how to efficiently store information about all such detours (unlike undirected graphs where we only store two shortest path trees) in $o(h_{st})$ space without significantly increasing the rounds needed for on-the-fly construction. We can have an on-the-fly approach storing $O(h_{st})$ values to include information about all detours, but in that case we might as well store the entire routing table.

\begin{theorem}
    \label{thm:undirrprecon}
    The undirected weighted RPaths algorithm can be modified to also compute routing tables for replacement path construction with $O(h_{st}+h_{rep})$ additional rounds. For undirected unweighted graphs, this is $O(D)$ additional rounds. 
    \begin{enumerate}
        \item The on-the-fly method uses $O(1)$ space per node. After edge failure, path construction takes $h_{st}+3h_{rep}$ rounds.
        \item The routing table method uses $O(h_{st})$ space per node. After edge failure, path construction takes $h_{st}+h_{rep}$ rounds.
    \end{enumerate}
\end{theorem}

\subsection{Minimum Weight Cycle Construction}
\label{sec:mwcrecon}

Our minimum weight cycle algorithms (Section~\ref{sec:mwcub}) can be modified to construct cycles using techniques similar to those used for replacement path construction. At each vertex on the MWC, we compute the next vertex on the cycle with $O(n)$ overhead, thus the total round complexity of the algorithm computing exact weights and constructing cycles is unchanged. Our main tool is to modify the APSP subroutine used in the MWC algorithms to compute the next vertex to $u$, $First(u,v)$, on any $u$-$v$ path, and use this to construct shortest paths that are part of a minimum weight cycle.

To reconstruct a minimum weight cycle through each node (ANSC), every vertex has a routing table of size up to $n$, with the entry corresponding to vertex $u$ representing the next vertex on a minimum cycle through $u$. 

We can also consider an on-the-fly model where we are given a vertex $u$ and need to construct a minimum weight cycle through $u$, with the advantage that we reduce the space used per node. Our on-the-fly algorithms require storing the $\Theta(n)$-sized routing table for APSP ($u$ stores $First(u,v)$ for each $v$);  this is a reasonable assumption since APSP routing tables are important information that can be used in other contexts as well. In addition to the routing table, we only need to store $O(1)$ additional information per node to construct MWC or ANSC on-the-fly.

\subsubsection{Directed Cycle Construction}
To construct a minimum weight cycle, each node in the cycle must determine the next vertex in the cycle. In the directed MWC algorithm (Section~\ref{sec:mwcub}), we compute a minimum weight cycle formed by a $u$-$v$ shortest path along with a $(v,u)$ edge. We can broadcast these vertices $u,v$ so that all vertices except $v$ set their next vertex to be the next vertex on a shortest path to $v$ and vertex $v$ sets its next vertex to $u$, and we construct a MWC in $O(D)$ additional rounds keeping the total $\tilde{O}(n)$ round complexity unchanged. We similarly construct cycles from the ANSC algorithm, and populate routing tables. We would need to broadcast $n$ pairs of $u,v$ vertices to construct cycles through all $n$ nodes, which is done in $O(n)$ rounds, keeping the total round complexity unchanged.

For the on-the-fly model, as mentioned earlier we assume that the routing table for APSP is stored at each node. In addition, the vertex $u$ stores only the vertex $v$ which is the last vertex in a minimum weight cycle through $u$. To construct the cycle through $u$, $u$ broadcasts the value of $v$ to all vertices, and other vertices use the APSP routing table to find the next vertex on a shortest path to $v$. This allows us to use $O(1)$ additional space per vertex, and constructs the cycle in $h_{cyc}$ rounds (same as routing table), where $h_{cyc}$ is the number of hops in the cycle.

\subsubsection{Undirected Cycle Construction}
Undirected cycle construction is quite similar to the directed case, with the change that a minimum weight cycle is composed of two shortest paths $P_{vu}, P_{v'u}$ and edge $(v,v')$ for some vertices $u,v,v'$. These vertices are computed in the undirected MWC algorithm, and can be broadcast at the end of the algorithm. The APSP information is then used to construct shortest paths $P_{vu}$ and $P_{v'u}$. The same argument applies to ANSC construction, and we broadcast $n$ tuples $(u,v,v')$ encoding information about cycles through all $n$ nodes in $O(n)$ rounds. The total round complexity of $\tilde{O}(n)$ is unchanged. In the on-the-fly model, the vertex $u$ stores $v,v'$ along with the APSP routing tables to construct a minimum weight cycle through $u$. So, we use $O(1)$ space per vertex in addition to the APSP tables, and constructing an MWC takes $h_{cyc}$ rounds (same as routing table).

%% file: conclusion.tex
\section{Conclusion and Further Research} 
\label{conclusion}
We have presented several upper and lower bounds on the round complexity of RPaths, 2-SiSP, MWC and ANSC in the CONGEST model. Many of our bounds are close to optimal, but we have some avenues for further research:
\begin{itemize}
\item Can we narrow or close the gap between the $\tilde{O}(n^{2/3}+\sqrt{nh_{st}}+D)$ upper bound and $\tilde{\Omega}(\sqrt{n}+D)$ lower bound for directed unweighted RPaths and approximate weighted directed RPaths?
\item Can the dependence on $h_{st}$ for RPaths be reduced? 
\item Round complexity of computing exact girth remains an open problem.
\item There are gaps between upper and lower bounds for MWC approximation in directed or weighted graphs; we make progress in a recent paper~\cite{mwcarxiv, mwc2024}, but can the gaps be reduced further?
\end{itemize}

%% file: rpaths-main.bbl
\begin{thebibliography}{10}

\bibitem{abboud2016}
Amir Abboud, Keren Censor-Hillel, and Seri Khoury.
\newblock Near-linear lower bounds for distributed distance computations, even
  in sparse networks.
\newblock In {\em International Symposium on Distributed Computing}, pages
  29--42. Springer, 2016.

\bibitem{abboud2015subcubic}
Amir Abboud, Fabrizio Grandoni, and Virginia~Vassilevska Williams.
\newblock Subcubic equivalences between graph centrality problems, apsp and
  diameter.
\newblock In {\em Proceedings of the Twenty-Sixth Annual ACM-SIAM Symposium on
  Discrete Algorithms}, SODA '15, page 1681–1697, USA, 2015. Society for
  Industrial and Applied Mathematics.

\bibitem{agarwal2018finegrained}
Udit Agarwal and Vijaya Ramachandran.
\newblock Fine-grained complexity for sparse graphs.
\newblock In {\em Proceedings of the 50th Annual {ACM} {SIGACT} Symposium on
  Theory of Computing, {STOC} 2018}, pages 239--252, Los Angeles, CA, USA,
  2018. {ACM}.

\bibitem{agarwal2020deterministic}
Udit Agarwal and Vijaya Ramachandran.
\newblock Faster deterministic all pairs shortest paths in {CONGEST} model.
\newblock In {\em {SPAA} 2020}, pages 11--21. {ACM}, 2020.

\bibitem{ancona2020}
Bertie Ancona, Keren Censor{-}Hillel, Mina Dalirrooyfard, Yuval Efron, and
  Virginia~Vassilevska Williams.
\newblock Distributed distance approximation.
\newblock In {\em 24th International Conference on Principles of Distributed
  Systems, {OPODIS} 2020}, volume 184 of {\em LIPIcs}, pages 30:1--30:17,
  Strasbourg, France (Virtual Conference), 2020. Schloss Dagstuhl -
  Leibniz-Zentrum f{\"{u}}r Informatik.

\bibitem{baryossef2002}
Z.~Bar-Yossef, T.S. Jayram, R.~Kumar, and D.~Sivakumar.
\newblock An information statistics approach to data stream and communication
  complexity.
\newblock In {\em The 43rd Annual IEEE Symposium on Foundations of Computer
  Science, 2002. Proceedings.}, pages 209--218, 2002.

\bibitem{bernsteinapsp}
Aaron Bernstein and Danupon Nanongkai.
\newblock Distributed exact weighted all-pairs shortest paths in near-linear
  time.
\newblock In {\em Proceedings of the 51st Annual {ACM} {SIGACT} Symposium on
  Theory of Computing, {STOC} 2019}, pages 334--342, Phoenix, AZ, USA, 2019.
  {ACM}.

\bibitem{bodwin2023restorable}
Greg Bodwin and Merav Parter.
\newblock Restorable shortest path tiebreaking for edge-faulty graphs.
\newblock {\em Journal of the ACM}, 70(5):1--24, 2023.

\bibitem{bremler2001restoration}
Anat Bremler-Barr, Yehuda Afek, Haim Kaplan, Edith Cohen, and Michael Merritt.
\newblock Restoration by path concatenation: Fast recovery of mpls paths.
\newblock In {\em Proceedings of the twentieth annual ACM symposium on
  Principles of Distributed Computing}, pages 43--52, 2001.

\bibitem{cao2023sssp}
Nairen Cao and Jeremy~T. Fineman.
\newblock Parallel exact shortest paths in almost linear work and square root
  depth.
\newblock In {\em Proceedings of the 2023 {ACM-SIAM} Symposium on Discrete
  Algorithms, {SODA} 2023}, pages 4354--4372, Florence, Italy, 2023. {SIAM}.

\bibitem{cao2021approximatesssp}
Nairen Cao, Jeremy~T. Fineman, and Katina Russell.
\newblock Brief announcement: An improved distributed approximate single source
  shortest paths algorithm.
\newblock In {\em Proceedings of the 2021 ACM Symposium on Principles of
  Distributed Computing}, PODC'21, page 493–496, New York, NY, USA, 2021.
  Association for Computing Machinery.

\bibitem{censorhillel2020girth}
Keren Censor{-}Hillel, Orr Fischer, Tzlil Gonen, Fran{\c{c}}ois~Le Gall, Dean
  Leitersdorf, and Rotem Oshman.
\newblock Fast distributed algorithms for girth, cycles and small subgraphs.
\newblock In {\em 34th International Symposium on Distributed Computing, {DISC}
  2020}, volume 179 of {\em LIPIcs}, pages 33:1--33:17, Virtual Conference,
  2020. Schloss Dagstuhl - Leibniz-Zentrum f{\"{u}}r Informatik.

\bibitem{chang2021near}
Yi-Jun Chang, Seth Pettie, Thatchaphol Saranurak, and Hengjie Zhang.
\newblock Near-optimal distributed triangle enumeration via expander
  decompositions.
\newblock {\em Journal of the ACM (JACM)}, 68(3):1--36, 2021.

\bibitem{chechik2019ssrp}
Shiri Chechik and Sarel Cohen.
\newblock Near optimal algorithms for the single source replacement paths
  problem.
\newblock In {\em Proceedings of the Thirtieth Annual ACM-SIAM Symposium on
  Discrete Algorithms}, pages 2090--2109. SIAM, 2019.

\bibitem{chechik2020ssrp}
Shiri Chechik and Ofer Magen.
\newblock Near optimal algorithm for the directed single source replacement
  paths problem.
\newblock In {\em 47th International Colloquium on Automata, Languages, and
  Programming, {ICALP} 2020, July 8-11, 2020, Saarbr{\"{u}}cken, Germany
  (Virtual Conference)}, volume 168, pages 81:1--81:17, 2020.

\bibitem{chechiksssp}
Shiri Chechik and Doron Mukhtar.
\newblock Single-source shortest paths in the {CONGEST} model with improved
  bounds.
\newblock {\em Distributed Comput.}, 35(4):357--374, 2022.

\bibitem{dinitz2020spanner}
Michael Dinitz and Caleb Robelle.
\newblock Efficient and simple algorithms for fault-tolerant spanners.
\newblock In {\em Proceedings of the 39th Symposium on Principles of
  Distributed Computing}, pages 493--500, 2020.

\bibitem{drucker2014}
Andrew Drucker, Fabian Kuhn, and Rotem Oshman.
\newblock On the power of the congested clique model.
\newblock In {\em {ACM} Symposium on Principles of Distributed Computing,
  {PODC} '144}, pages 367--376, Paris, France, 2014. {ACM}.

\bibitem{eden2021sublinear}
Talya Eden, Nimrod Fiat, Orr Fischer, Fabian Kuhn, and Rotem Oshman.
\newblock Sublinear-time distributed algorithms for detecting small cliques and
  even cycles.
\newblock {\em Distributed Comput.}, 35(3):207--234, 2022.

\bibitem{elkin2006mst}
Michael Elkin.
\newblock An unconditional lower bound on the time-approximation trade-off for
  the distributed minimum spanning tree problem.
\newblock {\em SIAM Journal on Computing}, 36(2):433--456, 2006.

\bibitem{forster2018sssp}
Sebastian Forster and Danupon Nanongkai.
\newblock A faster distributed single-source shortest paths algorithm.
\newblock In {\em 59th {IEEE} Annual Symposium on Foundations of Computer
  Science, {FOCS} 2018}, pages 686--697, Paris, France, 2018. {IEEE} Computer
  Society.

\bibitem{fraigniaud2024evencycle}
Pierre Fraigniaud, Mael Luce, Frederic Magniez, and Ioan Todinca.
\newblock Even-cycle detection in the randomized and quantum {CONGEST} model.
\newblock Technical report, arXiv, 2024.
\newblock \href {https://arxiv.org/abs/2402.12018} {\path{arXiv:2402.12018}}.

\bibitem{frischknecht2012}
Silvio Frischknecht, Stephan Holzer, and Roger Wattenhofer.
\newblock Networks cannot compute their diameter in sublinear time.
\newblock In {\em Proceedings of the Twenty-Third Annual {ACM-SIAM} Symposium
  on Discrete Algorithms, {SODA} 2012}, pages 1150--1162, Kyoto, Japan, 2012.
  {SIAM}.

\bibitem{ghaffarischeduling}
Mohsen Ghaffari.
\newblock Near-optimal scheduling of distributed algorithms.
\newblock In {\em Proceedings of the 2015 {ACM} Symposium on Principles of
  Distributed Computing, {PODC} 2015}, pages 3--12, Donostia-San
  Sebasti{\'{a}}n, Spain,, 2015. {ACM}.

\bibitem{ghaffari2016fault}
Mohsen Ghaffari and Merav Parter.
\newblock Near-optimal distributed algorithms for fault-tolerant tree
  structures.
\newblock In {\em Proceedings of the 28th ACM Symposium on Parallelism in
  Algorithms and Architectures}, pages 387--396, 2016.

\bibitem{ghaffari2015reach}
Mohsen Ghaffari and Rajan Udwani.
\newblock Brief announcement: Distributed single-source reachability.
\newblock In {\em Proceedings of the 2015 {ACM} Symposium on Principles of
  Distributed Computing, {PODC} 2015}, pages 163--165, Donostia-San
  Sebasti{\'{a}}n, Spain, 2015. {ACM}.

\bibitem{hoang2019round}
Loc Hoang, Matteo Pontecorvi, Roshan Dathathri, Gurbinder Gill, Bozhi You,
  Keshav Pingali, and Vijaya Ramachandran.
\newblock A round-efficient distributed betweenness centrality algorithm.
\newblock In {\em Proceedings of the 24th {ACM} {SIGPLAN} Symposium on
  Principles and Practice of Parallel Programming, PPoPP 2019}, pages 272--286,
  Washington, DC, USA, 2019. {ACM}.

\bibitem{holzer2012apsp}
Stephan Holzer and Roger Wattenhofer.
\newblock Optimal distributed all pairs shortest paths and applications.
\newblock In {\em {ACM} Symposium on Principles of Distributed Computing,
  {PODC} '12}, pages 355--364, Madeira, Portugal, 2012. {ACM}.

\bibitem{izumi2017triangle}
Taisuke Izumi and Fran{\c{c}}ois~Le Gall.
\newblock Triangle finding and listing in {CONGEST} networks.
\newblock In {\em Proceedings of the {ACM} Symposium on Principles of
  Distributed Computing, {PODC} 2017}, pages 381--389, Washington, DC, USA,
  2017. {ACM}.

\bibitem{katoh1982efficient}
Naoki Katoh, Toshihide Ibaraki, and Hisashi Mine.
\newblock An efficient algorithm for k shortest simple paths.
\newblock {\em Networks}, 12(4):411--427, 1982.

\bibitem{korhonen2017}
Janne~H. Korhonen and Joel Rybicki.
\newblock Deterministic subgraph detection in broadcast {CONGEST}.
\newblock In {\em 21st International Conference on Principles of Distributed
  Systems, {OPODIS} 2017}, volume~95 of {\em LIPIcs}, pages 4:1--4:16, Lisbon,
  Portugal, 2017. Schloss Dagstuhl - Leibniz-Zentrum f{\"{u}}r Informatik.

\bibitem{kushilevitzcomm}
Eyal Kushilevitz and Noam Nisan.
\newblock {\em Communication Complexity}.
\newblock Cambridge University Press, Cambridge, 1996.
\newblock \href {https://doi.org/10.1017/CBO9780511574948}
  {\path{doi:10.1017/CBO9780511574948}}.

\bibitem{lawler1972procedure}
Eugene~L Lawler.
\newblock A procedure for computing the k best solutions to discrete
  optimization problems and its application to the shortest path problem.
\newblock {\em Management science}, 18(7):401--405, 1972.

\bibitem{lenzen2019distributed}
Christoph Lenzen, Boaz Patt-Shamir, and David Peleg.
\newblock Distributed distance computation and routing with small messages.
\newblock {\em Distributed Computing}, 32(2):133--157, 2019.

\bibitem{mwcarxiv}
Vignesh Manoharan and Vijaya Ramachandran.
\newblock Improved approximation bounds for minimum weight cycle in the
  {CONGEST} model.
\newblock {\em arXiv preprint}, 2023.
\newblock URL: \url{https://arxiv.org/abs/2308.08670}.

\bibitem{mwc2024}
Vignesh Manoharan and Vijaya Ramachandran.
\newblock Computing minimum weight cycle in the {CONGEST} model.
\newblock In {\em ACM Symposium on Principles of Distributed Computing (PODC
  2024)}, page (to appear), 2024.

\bibitem{rp2024}
Vignesh Manoharan and Vijaya Ramachandran.
\newblock Computing replacement paths in the {CONGEST} model.
\newblock In {\em Structural Information and Communication Complexity (SIROCCO
  2024)}, page (to appear), 2024.

\bibitem{nanongkai2014approx}
Danupon Nanongkai.
\newblock Distributed approximation algorithms for weighted shortest paths.
\newblock In {\em Symposium on Theory of Computing, {STOC} 2014}, pages
  565--573, New York, NY, USA, 2014. {ACM}.

\bibitem{parter2020fault}
Merav Parter.
\newblock Distributed constructions of dual-failure fault-tolerant distance
  preservers.
\newblock In {\em 34th International Symposium on Distributed Computing (DISC
  2020)}. Schloss Dagstuhl-Leibniz-Zentrum f{\"u}r Informatik, 2020.

\bibitem{parter2022vertex}
Merav Parter.
\newblock Nearly optimal vertex fault-tolerant spanners in optimal time:
  sequential, distributed, and parallel.
\newblock In {\em Proceedings of the 54th Annual ACM SIGACT Symposium on Theory
  of Computing}, pages 1080--1092, 2022.

\bibitem{peleg2000distributed}
David Peleg.
\newblock {\em Distributed computing: a locality-sensitive approach}.
\newblock SIAM, USA, 2000.

\bibitem{peleg2013girth}
David Peleg, Liam Roditty, and Elad Tal.
\newblock Distributed algorithms for network diameter and girth.
\newblock In {\em Automata, Languages, and Programming - 39th International
  Colloquium, {ICALP} 2012}, volume 7392 of {\em Lecture Notes in Computer
  Science}, pages 660--672, Warwick, UK, 2012. Springer.

\bibitem{peleg2000mst}
David Peleg and Vitaly Rubinovich.
\newblock A near-tight lower bound on the time complexity of distributed
  minimum-weight spanning tree construction.
\newblock {\em SIAM Journal on Computing}, 30(5):1427--1442, 2000.

\bibitem{Pettie2022}
Seth Pettie.
\newblock Personal communication, {Jan} 2022.

\bibitem{razborov1992}
A.~A. Razborov.
\newblock On the distributional complexity of disjointness.
\newblock {\em Theor. Comput. Sci.}, 106(2):385--390, {Dec.} 1992.

\bibitem{roditty2012replacement}
Liam Roditty and Uri Zwick.
\newblock Replacement paths and k simple shortest paths in unweighted directed
  graphs.
\newblock {\em ACM Transactions on Algorithms (TALG)}, 8(4):1--11, 2012.

\bibitem{rozhon22sssp}
V{\'{a}}clav Rozhon, Bernhard Haeupler, Anders Martinsson, Christoph Grunau,
  and Goran Zuzic.
\newblock Parallel breadth-first search and exact shortest paths and stronger
  notions for approximate distances.
\newblock In {\em Proceedings of the 55th Annual {ACM} Symposium on Theory of
  Computing, {STOC} 2023}, pages 321--334, Orlando, FL, USA, 2023. {ACM}.

\bibitem{sarma2012distributed}
Atish~Das Sarma, Stephan Holzer, Liah Kor, Amos Korman, Danupon Nanongkai,
  Gopal Pandurangan, David Peleg, and Roger Wattenhofer.
\newblock Distributed verification and hardness of distributed approximation.
\newblock {\em SIAM Journal on Computing}, 41(5):1235--1265, 2012.

\bibitem{williams2010subcubic}
Virginia~Vassilevska Williams and R.~Ryan Williams.
\newblock Subcubic equivalences between path, matrix, and triangle problems.
\newblock {\em J. {ACM}}, 65(5):27:1--27:38, 2018.

\bibitem{yen1971finding}
Jin~Y Yen.
\newblock Finding the k shortest loopless paths in a network.
\newblock {\em Management science}, 17(11):712--716, 1971.

\end{thebibliography}
